\documentclass[11pt,a4paper]{article}
\usepackage[margin=2.2cm]{geometry}
\usepackage{tikz,graphicx} % Required for inserting images
\usepackage{amsmath,amsthm,bbm,mathrsfs}
\usepackage{amssymb}
\usepackage{xcolor,todonotes,soul}
\usepackage[pdfpagelabels]{hyperref}
\hypersetup{
	breaklinks=true,   % splits links across lines
	colorlinks=true,   % displays links as colored text instead of blocks
	pdfusetitle=true,  % \title and \author values into pdf metadata
}
 
\usepackage{comment}

\newcommand{\even}{{\textnormal{even}}}
\newcommand{\reg}{{\textnormal{reg}}}
\DeclareMathOperator{\Tr}{Tr}

\newcommand{\FS}{\textnormal{FS}}
\newcommand{\FW}{\textnormal{FW}}
\newtheorem{theorem}{Theorem}[section]
\newtheorem{definition}[theorem]{Definition}
\newtheorem{proposition}[theorem]{Proposition}
\newtheorem{lemma}[theorem]{Lemma}
\numberwithin{equation}{section}

\usepackage{scalerel}
\usetikzlibrary{svg.path}

\definecolor{orcidlogocol}{HTML}{A6CE39}
\tikzset{
	orcidlogo/.pic={
		\fill[orcidlogocol] svg{M256,128c0,70.7-57.3,128-128,128C57.3,256,0,198.7,0,128C0,57.3,57.3,0,128,0C198.7,0,256,57.3,256,128z};
		\fill[white] svg{M86.3,186.2H70.9V79.1h15.4v48.4V186.2z}
		svg{M108.9,79.1h41.6c39.6,0,57,28.3,57,53.6c0,27.5-21.5,53.6-56.8,53.6h-41.8V79.1z M124.3,172.4h24.5c34.9,0,42.9-26.5,42.9-39.7c0-21.5-13.7-39.7-43.7-39.7h-23.7V172.4z}
		svg{M88.7,56.8c0,5.5-4.5,10.1-10.1,10.1c-5.6,0-10.1-4.6-10.1-10.1c0-5.6,4.5-10.1,10.1-10.1C84.2,46.7,88.7,51.3,88.7,56.8z};
	}
}

\newcommand\orcidicon[1]{\href{https://orcid.org/#1}{\mbox{\scalerel*{
				\begin{tikzpicture}[yscale=-1,transform shape]
					\pic{orcidlogo};
				\end{tikzpicture}
			}{|}}}}

\newcommand{\xb}{\boldsymbol{x}}
\newcommand{\dd}{\mathrm{d}}
\newcommand{\ee}{\mathrm{e}}
\newcommand{\ii}{\mathrm{i}}
\newcommand{\RR}{\mathbb{R}}
\newcommand{\NN}{\mathbb{N}}
\newcommand{\CC}{\mathbb{C}}

\DeclareMathOperator{\diag}{diag}

\newcommand{\DD}{\mathscr{D}}

\newcommand{\Gc}{\mathcal{G}}
\newcommand{\Hc}{\mathcal{H}}
\newcommand{\Jc}{\mathcal{J}}
\newcommand{\Kc}{\mathcal{K}}
\newcommand{\Lc}{\mathcal{L}}
\newcommand{\Tc}{\mathcal{T}}
\newcommand{\Wc}{\mathcal{W}}
\newcommand{\Lor}{\mathfrak{L}}

\DeclareMathOperator{\SO}{SO}

\DeclareMathOperator{\Hess}{Hess}
\DeclareMathOperator{\expc}{expc}
\DeclareMathOperator{\sinc}{sinc}
\DeclareMathOperator{\sinhc}{sinhc}

\title{Quantum Energy Inequalities with a twist:\\ rotating frames along stationary worldlines}
\author{Christopher J. Fewster\,\orcidicon{0000-0001-8915-5321}\thanks{\tt chris.fewster@york.ac.uk}
~and Daniel Lomax Ford\,\orcidicon{0009-0002-8210-5633}\thanks{\tt dblf500@york.ac.uk}\\ 
\emph{\small Department of Mathematics, Ian Wand Building, Deramore Lane,}\\
\emph{\small University of York, York YO10 5GH, UK}
 }
\date{\today}

\begin{document}

\maketitle

\begin{abstract}
The energy densities of the quantised scalar and electromagnetic fields need not be positive, but
their averages along timelike worldlines are constrained by Quantum Energy Inequalities
(QEIs) that hold in flat and curved spacetimes for arbitrary Hadamard quantum field states. However the general QEI bound is not easy to
compute in closed form. Here, we compute closed form expressions for the massless scalar and electromagnetic fields along stationary worldlines in Minkowski spacetime. These worldlines are solutions to the four-dimensional Frenet--Serret equations with constant curvature invariants and include inertial, uniformly linearly accelerated, and uniformly rotating worldlines as special cases. 
The QEI bound depends on a choice of tetrad along the worldline, and we compute it for any tetrad with constant angular
velocity relative to the Frenet--Serret tetrad. Optimising over the angular velocity, we show that the tightest electromagnetic QEI bound occurs when the tetrad is invariant under Fermi--Walker tetrads. Surprisingly, the Fermi--Walker tetrad does not produce the tightest QEI bound for massless scalar fields, and we investigate the optimal bound both analytically and numerically. 
Our results and methodology significantly improve on those in [CJ Fewster \& J Thompson, Class. Quantum Grav. {\bf 40} (2023) 175008],
in which only Frenet--Serret tetrads were considered for massless scalar fields. 
\end{abstract}

\section{Introduction}
 
Typical classical matter models obey the Weak Energy Condition (WEC): all observers register an everywhere nonnegative energy density. In particular, the WEC holds for  the electromagnetic and minimally coupled scalar fields. Similar energy conditions are often assumed on unknown classical matter models as, for example, in the singularity theorems of general relativity~\cite{Penrose:1965,Hawking:1966}. However, no quantum field theory (QFT) can satisfy any local energy condition~\cite{EpsteinGlaserJaffe:1965}, and under mild assumptions on the theory it has been shown that quantum states can be tuned to give an arbitrarily expectation value of energy density at a given point \cite{Fews05}. These violations of classical energy conditions are of particular interest as they could allow for a range of phenomena from violating the second law of thermodynamics \cite{Ford1978QuantumCE} to the ability to construct wormholes \cite{MorrisThorne:1988} and warp drives \cite{Alcubierre_1994}. Such possibilities motivate the investigation of what restrictions might replace the classical energy conditions in QFT.  These restrictions are known as quantum energy inequalities (QEIs), and place lower bounds on the expectation value of the average energy density of a quantum field. The simplest QEIs are independent of quantum state, and constrain the magnitude and duration of negative energy densities, or negative energy fluxes. They limit attempts to use quantum fields to find candidates for exotic matter that are required to construct exotic spacetimes \cite{FewsterRoman:2005,FordRoman:1996worm,PfenningFord:1997}, and have been used to give versions of the singularity theorems in semiclassical gravity~\cite{FewsterKontou2022}. See~\cite{Fewster2017} and~\cite{KontouSanders:2020} for recent reviews of classical energy conditions and QEIs.

Specific QEIs have been derived for a variety of QFTs, starting from works of Ford and Roman~\cite{Ford1978QuantumCE,Ford:1991,FordRoman:1995}. For example, in $4$-dimensional Minkowski spacetime, the averaged energy density of the massless scalar field and electromagnetic field along an inertial curve $\gamma(s)$
obey inequalities
\begin{equation}\label{eq:sample_QEIs}
 \int_{-\infty}^{\infty}\dd s\rvert g(s) \rvert ^{2} \langle :T_{\mu\nu}\dot{\gamma}^{\mu}\dot{\gamma}^{\nu}: \rangle _{\omega} (\gamma(s))    \ge - \frac{N}{16\pi^2}\int_{-\infty}^{\infty} \dd s\, |g''(s)|^2, 
\end{equation}
for any smooth function $g$ of sufficiently rapid decay,
where $N=1$ for the scalar field~\cite{BoundNegEnDenFewsterEveson1998} and $N=2$ for the electromagnetic field~\cite{Pfenning:2002}, and assuming the inertial curve is parametrised by proper time. With $N=1$, the QEI~\eqref{eq:sample_QEIs} also applies to massive scalar fields, although tighter mass-dependent bounds are available~\cite{BoundNegEnDenFewsterEveson1998,FewsterEveson:2007}. The overall factor of $2$ for electromagnetism is explained by its two independent polarisations.

The left-hand side of~\eqref{eq:sample_QEIs} is independent of the state $\omega$ and its simple form makes it straightforward to interpret. For instance, replacing $g$ by a scaled version $g_\lambda(s)=\lambda^{-1/2}g(s/\lambda)$ gives
\begin{equation}
 \int_{-\infty}^{\infty}\dd s\rvert g_\lambda(s) \rvert ^{2} \langle :T_{\mu\nu}\dot{\gamma}^{\mu}\dot{\gamma}^{\nu}: \rangle _{\omega} (\gamma(s))    \ge - \frac{N\|g''\|^2}{16\pi^2\lambda^4}, 
\end{equation}
where $\|\cdot\|$ denotes the $L^2$-norm, $\|g\|^2=\int_{-\infty}^\infty \dd s\,|g(s)|^2$, which shows how the QEI constrains the magnitude and (fourth power of the) duration of negative energy density in a manner reminiscent of the uncertainty principle. As $\lambda\to 0$, the lower bound tends to $-\infty$, as it must because pointwise energy densities are unbounded below; as $\lambda\to\infty$, on the other hand, the bound tends to zero, showing that a static negative energy density cannot be maintained along an inertial worldline by any state on Minkowski spacetime.

Generalisations of the QEIs to arbitrary timelike trajectories in general curved spacetimes have been proved, starting with~\cite{GenQEI2000CFewster} which established 
a mathematically rigorous QEI for the minimally coupled free scalar field of mass $m\ge 0$ for averages along any smooth timelike curve $\gamma$ in any globally hyperbolic spacetime (of arbitrary dimension).  This bound takes the form
\begin{equation}\label{eq:QEI_point-split}
    \int_{-\infty}^{\infty}\dd s\,\rvert g(s) \rvert ^{2} \langle :T_{\mu\nu}\dot{\gamma}^{\mu}\dot{\gamma}^{\nu}: \rangle _{\omega} (\gamma(s)) \geq -\int_{0}^{\infty}\frac{\dd\alpha}{\pi}[(g\otimes g)T]^\wedge(-\alpha,\alpha) > -\infty,
\end{equation}
where $\gamma(s)$ is parametrised by proper time $s$, $g$ is any real-valued smooth compactly supported function and the hat indicates Fourier transform (conventions are given later). The colons on the left-hand side denote normal ordering relative to a Hadamard reference state which is also used to construct the point-split energy density $T$ as a bi-distribution on the curve, as is explained in Section~\ref{sec:QEI_background}. The inequality then holds for any Hadamard state $\omega$ of the field. A similar bound was established for the electromagnetic field (and also the massive spin-$1$ Proca field) in~\cite{FewsterPfenning:2003}.  

It is clear that the bound in~\eqref{eq:QEI_point-split} is much harder to compute than~\eqref{eq:sample_QEIs}: one needs to be able to find a reference state for which the point-split energy density and its Fourier transform may be evaluated, before taking the final integral. The point-split energy density $T$ also depends on a choice of tetrad in a neighbourhood of $\gamma$, and the impact of this choice on the QEI bound has not been investigated in detail.  Accordingly~\eqref{eq:QEI_point-split} is harder to interpret than~\eqref{eq:sample_QEIs},
to which it reduces for inertial worldlines in Minkowski spacetime, using the Minkowski vacuum as the reference state~\cite{GenQEI2000CFewster}. Similar reductions to other known results (e.g.~\cite{FewsterTeo:1999}) were obtained in stationary spacetimes, using the stationary ground state as the reference, because the point-split energy density is translationally invariant, facilitating the computation of the Fourier transform.  

One way to extract further insight from~\eqref{eq:QEI_point-split} is to consider particular noninertial worldlines in Minkowski spacetime, using the vacuum state as the reference. In particular, we will consider \emph{stationary worldlines}~\cite{Kottler1912,Synge1967,Letaw:1981},
namely, timelike curves that solve the generalised Frenet--Serret equations with constant curvature invariants. In four-dimensional Minkowski spacetime there are three invariants: the curvature $\kappa$, measuring proper acceleration, and the torsion $\tau$ and hypertorsion~$\upsilon$ measuring proper angular velocity. Their values are constrained so that $\kappa$ and $\tau$ are nonnegative, and that $\upsilon$ vanishes if $\tau=0$, and $\tau$ and $\upsilon$ vanish if $\kappa=0$. These trajectories can be described as timelike Killing orbits or as curves whose velocity vectors evolve under a $1$-parameter subgroup of the Lorentz group. They fall into six types, including inertial motion, uniform linear acceleration, and uniform rotation, together with three less familiar types called the semicubical parabola, catenary and loxodromic worldlines (other names exist in the literature). Conveniently, the Frenet--Serret equations provide a tetrad along $\gamma$, which can be used to construct the point-split energy density. 

Stationary worldlines are of particular interest in QFT, partly because they are sufficiently simple to allow closed form expressions for various quantities, but also for their physical significance. For example,
the Unruh effect describes the thermal excitation of small probe that travels along the uniformly accelerating trajectory \cite{DeBievreMerkli:2006,Unruh:1976} while interacting with a QFT in its Minkowski vacuum state. One way to understand this effect is that the Minkowski vacuum restricts to the Rindler wedge determined by the trajectory as a thermal state at the Unruh temperature. A related fact is that there is a Fulling ground state on the wedge in which the probe would not register excitations.

The Unruh effect has inspired research into the thermal excitation for observers following other stationary worldlines \cite{GoodJuarezAubryetal:2020,Letaw:1981}, and whether coordinates adapted to these worldlines are also associated with ground states following from restricting the Minkowski vacuum to thermal states, similar to the Rindler state \cite{LetawPfautsch:1981}. In particular, the uniformly rotating trajectory plays an important role in proposals to observe the Unruh effect, as circular motion is confined to a finite volume \cite{BellLeinaas:1983, Gooding_etal_PRL:2020}.

In~\cite{FewsterThompson2023} the QEI bound~\eqref{eq:QEI_point-split} along stationary worldlines was evaluated in closed form for the massless scalar field, using the Frenet--Serret tetrad. The result was given in terms of the curvature invariants as 
\begin{equation}\label{eq:scalar_QEI}
    \int_{-\infty}^{\infty}\dd s\rvert g(s) \rvert ^{2} \langle :T_{\mu\nu}\dot{\gamma}^{\mu}\dot{\gamma}^{\nu}: \rangle _{\omega} (\gamma(s)) \geq -\frac{1}{16\pi^{2}}\int_{-\infty}^{\infty}\dd s\,(|g''(s)|^{2}+2A|g'(s)|^{2}+B|g(s)|^{2}),
\end{equation}
where $A$ and $B$ are 
\begin{equation}\label{scalar_curvature_formula}
    A=\kappa^{2}+\tau^{2}+\upsilon^{2} ,\quad B=\frac{1}{90}(3\kappa^{4}+62\kappa^{2}\tau^{2}+30(\kappa^{2}+\tau^{2}+\upsilon^{2})^{2}).
\end{equation}
Note that $A$ and $B$ depend on all three curvatures, and $B$ grows quartically in all of them. If $g$ is replaced by the scaled version $g_\lambda$ then we find 
\begin{equation}\label{eq:scaled_QEI}
    \int_{-\infty}^{\infty}\dd s\rvert g_\lambda(s) \rvert ^{2} \langle :T_{\mu\nu}\dot{\gamma}^{\mu}\dot{\gamma}^{\nu}: \rangle _{\omega} (\gamma(s)) \geq -\frac{1}{16\pi^{2}}\left(
    \frac{\|g''\|^2}{\lambda^4} + 2A\frac{\|g'\|^2}{\lambda^2} + B\|g\|^2\right),
\end{equation}
with the same asymptotic behaviour as in the inertial case as $\lambda\to 0+$, a very natural result bearing in mind the equivalence principle. At intermediate scales this QEI permits departures from the inertial result, determined by the sums of the squares of the curvatures, while for $\lambda\to+\infty$ the bound tends to $-B\|g\|^2/(16\pi^2)$, which only vanishes in the inertial case $\kappa=\tau=\upsilon=0$. Thus the QEI permits, but does not guarantee, the existence of a constant negative energy density $-B/(16\pi^2)$ along the stationary worldlines. In the case of uniform linear acceleration, this value is achieved by the Fulling--Rindler vacuum state on the wedge generated by the trajectory, while the vanishing value for inertial trajectories is of course attained by the Minkowski vacuum. 
For other stationary curves, the QEI would allow the energy density could grow in magnitude with the fourth powers of $\tau$ and $\upsilon$. The question of whether there are states saturating the bounds for other stationary worldlines was left open in~\cite{FewsterThompson2023}.

A more detailed consequence of~\eqref{eq:scalar_QEI} is that it provides a lower bound on the maximum value of the energy density along any interval of bounded proper duration
\begin{equation}\label{eq:eigenproblem_estimate}
    \max_{s\in [s_1,s_2]} \langle :T_{\mu\nu}\dot{\gamma}^{\mu}\dot{\gamma}^{\nu}: \rangle _{\omega} (\gamma(s)) \ge -\frac{B}{16\pi^2}-\frac{\mu(AT^2)}{16\pi^2 T^4},
\end{equation}
where $T=s_2-s_1$. That is, the energy density cannot remain below the value on the right-hand side for more than a time $T$. Here, $\mu(a)$ is the minimum eigenvalue of a particular differential operator depending on $a$. Details of this argument, which is closely based on those in~\cite{FewsterTeo:2000,QEI2006FewsterPfenning}, can be found in Appendix~\ref{appx:eigenvalues}, where it is shown that $\mu(a)\to \mu(0)=500.56$ as $a\to 0$, that $\mu(a)$ increases monotonically with $a$ and that
$\mu(a)\sim 2\pi^2 a$ for $a\gg \pi^2$. Again, $B$ determines the long time $T\to\infty$ limit of~\eqref{eq:eigenproblem_estimate}, while $A^{-1/2}$ is a timescale distinguishing the short-term regime, in which
\begin{equation}
    \max_{s\in [s_1,s_2]} \langle :T_{\mu\nu}\dot{\gamma}^{\mu}\dot{\gamma}^{\nu}: \rangle _{\omega} (\gamma(s)) \gtrsim -\frac{B}{16\pi^2}-\frac{\mu(0)}{16\pi^2 T^4}, 
    \qquad (T\ll A^{-1/2})
\end{equation}
from the long-term regime, in which one finds a modified power-law in $T$,
\begin{equation}
    \max_{s\in [s_1,s_2]} \langle :T_{\mu\nu}\dot{\gamma}^{\mu}\dot{\gamma}^{\nu}: \rangle _{\omega} (\gamma(s)) \gtrsim -\frac{B}{16\pi^2}-\frac{A}{8T^2}. 
    \qquad (T\gg A^{-1/2})
\end{equation} 
Thus the constants in the QEI bound~\eqref{eq:scalar_QEI} have a clear physical meaning and their value is important in determining the possible distribution of negative energy density along stationary worldlines. According to~\eqref{scalar_curvature_formula}, the timescale $A^{-1/2}$ depends on the sum of the squares of the curvatures, and the $T\to\infty$ limit grows quartically with $\tau$ and $\upsilon$ as well as $\kappa$. Increasing one or both of $\tau$, $\upsilon$ with fixed $\kappa$ and $T$ leads into the short-term regime and makes both terms in~\eqref{eq:eigenproblem_estimate} increasingly negative.
One of the purposes of this paper is to show that this conclusion is too pessimistic and that the bound~\eqref{eq:scalar_QEI} and~\eqref{scalar_curvature_formula} can be improved significantly (though they remain true).

More widely, we will develop the study of QEIs along stationary worldlines in three ways. First, we compute the QEI for the electromagnetic field, in order to understand how the spin of the field interacts with the twisting motion of stationary worldlines with nonzero torsion and hypertorsion (we thank Aron Wall for raising this interesting question). Second, we have been able to simplify the methodology of~\cite{FewsterThompson2023} by expressing many relevant quantities as functions of a matrix formed from the curvature invariants. In turn, this has permitted us to compute the QEI for both electromagnetism and the scalar field using any tetrad that is uniformly rotating with 
constant angular velocity $\boldsymbol{\varpi}$ relative to the Frenet--Serret tetrad,
in the general form
\begin{equation}\label{eq:overall_QEI}
    \int_{-\infty}^{\infty}\dd s\rvert g(s) \rvert ^{2} \langle :T_{\mu\nu}\dot{\gamma}^{\mu}\dot{\gamma}^{\nu}: \rangle _{\omega} (\gamma(s)) \geq -\frac{N}{16\pi^{2}}\int_{-\infty}^{\infty}\dd s\,(|g''(s)|^{2}+2A|g'(s)|^{2}+B|g(s)|^{2}),
\end{equation}
with $N=1$ for scalars and $N=2$ for electromagnetism as before, but with parameters $A$ and $B$ that now depend both on the type of field and on $\boldsymbol{\varpi}$. 
The expressions for $A$ and $B$ can be found in~\eqref{eq:A_scalar_Omega} and~\eqref{eq:B_scalar_Omega} for massless scalar fields and in~\eqref{eq:A_em_Omega} and~\eqref{eq:B_em_Omega} for electromagnetic fields.
By optimising over the angular velocity, we then determine the tightest bound in this class of QEIs. 

The optimisation is simplest for electromagnetism: we show that the global optimum is obtained from a tetrad that is invariant under Fermi--Walker transport along $\gamma$ (see, e.g., Section 6.5 of~\cite{misner2017gravitation}) and therefore represents the rods and clocks of a nonrotating observer travelling along the stationary worldline. This is a very natural tetrad, and indeed its use was advocated in~\cite{QEI2006FewsterPfenning} as a way to obtain a locally covariant QEI. The bound obtained for the Fermi--Walker tetrad has
\begin{equation}\label{eq:ABFW}
    A_{\textnormal{e.m.}}^\FW = \kappa^2, \qquad B_{\textnormal{e.m.}}^\FW= \frac{1}{90}\left( 33\kappa^4+32\kappa^2\tau^2\right).
\end{equation}
In particular, $A_{\textnormal{e.m.}}^\FW$ is independent of both torsion and hypertorsion, 
while $B_{\textnormal{e.m.}}^\FW$ is independent of hypertorsion and grows only quadratically with the torsion. Accordingly,  
the leading order departure from the inertial QEI at small scales is determined by the curvature alone, while \emph{constant} energy densities on stationary worldlines are bounded below by
\begin{equation}\label{eq:constantbound_em}
    \langle \rho\rangle \ge -\frac{B^\FW_{\textnormal{e.m.}}}{8\pi^2}=-\frac{33\kappa^4+32\kappa^2\tau^2}{720\pi^2}.
\end{equation}

This is a much stronger conclusion than would be reached if only the Frenet--Serret tetrad was considered, for which the QEI takes the form~\eqref{eq:scalar_QEI} with parameters
\begin{equation}
   A_{\textnormal{e.m.}}^\FS =A_{\textnormal{e.m.}}^\FW + 2(\tau^2+\upsilon^2), \qquad B_{\textnormal{e.m.}}^\FS=B_{\textnormal{e.m.}}^\FW + 
   \frac{2}{3}(2\kappa^2+\tau^2+\upsilon^2)(\tau^2+\upsilon^2), 
\end{equation} 
which (like the scalar case) depend on all curvatures with $B_{\textnormal{e.m.}}^\FS$
growing quartically in $\tau$ and $\upsilon$, and admits the possibility that there could be  states with constant energy density growing unboundedly in $\upsilon$; this possibility is ruled out by~\eqref{eq:ABFW}. 
 
For the massless scalar field, the QEI obtained using the Fermi--Walker tetrad yields the same parameters $A^\FW=A^\FW_\textnormal{e.m.}$ and $B^\FW=B^\FW_\textnormal{e.m.}$, which implies that the electromagnetic QEI is exactly twice the scalar one, if the Fermi--Walker tetrad is used in both cases. 
The values $A$ and $B$ from~\eqref{eq:ABFW} are both lower than the corresponding parameters~\eqref{scalar_curvature_formula} obtained in~\cite{FewsterThompson2023} using Frenet--Serret tetrads. Indeed, one has
\begin{equation}\label{eq:AB_FS}
    A^\FS = A^\FW + \tau^2 + \upsilon^2
    ,\qquad 
    B^\FS = B^\FW + \frac{1}{3}(3(\kappa\tau)^2+2(\kappa\upsilon)^2 +(\tau^2+\upsilon^2)^2) ,   
\end{equation}
so the QEIs obtained here are more stringent than those of~\cite{FewsterThompson2023}, 
replacing the quartic dependence of $B$ on $\tau$ and $\upsilon$ by quadratic dependence on $\tau$. Surprisingly, however, and in contrast with electromagnetism, it turns out that the Fermi--Walker tetrad does not minimise $B$ (even locally).  We will show that the global minimum can be written
\begin{equation}
    B_\textnormal{min} = \frac{\kappa^2}{90}(33\kappa^2+\beta(\tau/\kappa,\upsilon/\kappa)\tau^2),
\end{equation}
where $26<\beta(t,u)<32$ for all $(t,u)\in(0,\infty)\times\RR$. The lower bound is approached for $(t,u)\to 0$ in which regime one has
\begin{equation}
    \beta(t,u) = 26+\frac{6}{125}(t^2+25u^2) +  \mathcal{O}((t^2+u^2)^2),
\end{equation} 
and the minimum is achieved for a tetrad with angular velocity approximately $(\upsilon,0,4\tau/5)$ relative to the Frenet--Serret frame (while the Fermi--Walker frame has relative angular velocity $\boldsymbol{\varpi}^\FW=(\upsilon,0,\tau)$). On the other hand,
the upper bound is approached in various limits as $t^2+u^2\to\infty$, with the optimal angular velocity asymptotically equal to the Fermi--Walker $\boldsymbol{\varpi}^\FW$.

As a result, a constant energy density of the massless scalar field along a stationary worldline obeys
\begin{equation}\label{eq:constantbound}
    \langle \rho\rangle \ge -\frac{33\kappa^4+\beta(\tau/\kappa,\upsilon/\kappa)\kappa^2\tau^2}{1440\pi^2}.
\end{equation}
We caution that even once the QEI~\eqref{eq:QEI_point-split} is optimised over uniformly rotating tetrads, there is no guarantee that it is sharp. Thus it is unclear whether there are states with constant energy density approaching $-B^\FW_{\textnormal{e.m.}}/(8\pi^2)$ except in the inertial and uniform linear acceleration, for which the bound is saturated by the Minkowski vacuum and the Fulling--Rindler vacuum respectively as mentioned above. As a consistency check, we will compute the energy density of the electromagnetic field along a loxodromic worldline (the most general of the standard types) in the Fulling--Rindler vacuum on the wedge generated as the worldline's causal hull. It turns out that the energy density is constant, and satisfies~\eqref{eq:constantbound_em}, saturating it in the limit $\tau\to 0$ as expected.

We briefly describe the structure of the paper. In Section~\ref{sec:QEI_background} we recall some background information on QEIs, describing in particular how the point-split energy density $T$ is computed in terms of a tetrad defined along the observer's trajectory. This leads to expressions for $T$ in terms of smooth matrix-valued functions $C(s,s')$ and $D(s,s')$ that can be constructed from the worldline and tetrad. We also describe various simplifications to the QEI bound in the translationally-invariant situation where $T(s,s')=\Tc(s-s')$ for some distribution $\Tc$. In particular, an observation made in~\cite{FewsterThompson2023} can be used to extract 
the QEI bound from the short-distance behaviour of $\Tc$ without having to compute and then integrate the Fourier transform $\widehat{\Tc}$. Then, in Section~\ref{sec:stationary} we define and give the main properties of the stationary worldlines, including the 
spectral properties and functional calculus of a matrix formed form the curvatures. These results are used in Section~\ref{sec:summary} to compute the matrix functions $C$ and $D$ for tetrads with uniform angular velocity relative to the Frenet--Serret tetrad, leaving most of the details to the appendices. Section~\ref{sec:QEIs} gives the computation of the QEIs along stationary worldlines for both the massless minimally coupled scalar field and the electromagnetic field, and optimises the resulting bounds with respect to the tetrad angular velocity.  Section~\ref{sec:loxodromic} provides a consistency check on our QEIs by considering the energy density along loxodromic worldlines in a Fulling--Rindler state for a corresponding wedge region. 
Finally, we summarise and give some outlook in Section~\ref{sec:conclusion}. Four appendices give technical details deferred from the main text.

\paragraph{Units and conventions}

Units are employed in which $\hbar$, $c$ and the vacuum permittivity $\varepsilon_0$ are all set to unity. The Minkowski metric is $\eta=\dd t^2-\dd x^2-\dd y^2-\dd z^2$ in standard inertial coordinates $(t,x,y,z)$. Greek and Latin alphabets are used for tensor and frame indices respectively.  The analytic functions $\sinc$, $\sinhc$ and $\expc$ are defined by 
\begin{equation}\label{eq:sin(h)c}
    \sinc z =\frac{\sin z}{z}, \qquad \sinhc z= \frac{\sinh z}{z},
    \qquad \expc(z) = \frac{\exp(z)-1}{z}
\end{equation}
for $z\neq 0$, and $\sinc(0)=\sinhc(0)=\expc(0)=1$. 
More generally, the \emph{cardinal function} of an entire analytic function $f$ is the entire analytic function $f_\textnormal{c}$ defined by 
\begin{equation} 
f_\textnormal{c}(z)=\begin{cases} (f(z)-f(0))/z & z\neq 0\\ f'(0) & z=0; \end{cases}
\end{equation}
we also use the same notation if $f\in C^\infty(\RR)$, obtaining $f_\textnormal{c}\in C^\infty(\RR)$.
The convention used for the Fourier transform is 
\begin{equation}
\widehat{u}(\alpha) =\int_{-\infty}^\infty \dd s\,  \ee^{\ii\alpha s}u(s).
\end{equation}

\section{QEIs for the scalar and electromagnetic fields}
\label{sec:QEI_background}
\subsection{Massless scalar field}

The massless scalar field, with Lagrangian $\mathcal{L}=\tfrac{1}{2}(\nabla_\mu\phi)\nabla^\mu\phi$, has classical stress-energy tensor
\begin{equation}
    T_{\mu\nu} = (\nabla_\mu\phi)\nabla_\nu\phi - \frac{1}{2}\eta_{\mu\nu} (\nabla^\lambda\phi)\nabla_\lambda\phi
\end{equation}  
in Minkowski space. It satisfies the weak energy condition: if $u^\mu$ is a timelike unit vector at some spacetime point $x$, let $e^\mu_a$ be any tetrad at $x$ such that $e_0^\mu=u^\mu$. Then the energy density seen by an observer at $x$ with $4$-velocity $u^\mu$ may be written
\begin{equation}
      T_{\mu\nu}u^\mu u^\nu= \frac{1}{2}\sum_{a=0}^3 (e_a^\mu \phi)^2,
\end{equation}
which is manifestly nonnegative. 

The quantized field has vacuum two-point function
\begin{equation}
    \langle \phi(x)\phi(x')\rangle = G_0(x,x'),
\end{equation}
where, for $x=(t,\xb)$, $x'=(t',\xb')$, 
\begin{equation}\label{scalar_2-point}
    G_0(x,x')=\lim_{\epsilon\rightarrow 0+}\frac{-1}{4\pi^2\sigma_\epsilon(x,x')}, \qquad  \sigma_\epsilon(x,x')=(t-t'-\ii\epsilon)^2-\|\xb-\xb'\|^2,
\end{equation}
and the limit is understood in a distributional sense. 
The quantised theory does not obey the weak energy condition. Instead, in~\cite{GenQEI2000CFewster}, it was shown that averages of the normal ordered energy density along
a smooth timelike curve $\gamma$ obey a QEI of the form
\begin{equation}\label{eq:QEI}
    \int_{-\infty}^{\infty}\dd s\,g(s)^{2} \langle :T_{\mu\nu}\dot{\gamma}^\mu\dot{\gamma}^\nu: \rangle _{\omega} (\gamma(s)) \geq -\int_{0}^{\infty}\frac{\dd\alpha}{\pi}[(g\otimes g)T]^{\wedge}(-\alpha,\alpha) > -\infty
\end{equation}
for all real-valued test functions $g\in C^{\infty}_{0}(\mathbb{R})$ and all Hadamard states $\omega$, where 
\begin{equation}\label{eq:pointsplit_T_scalar}
    T(s,s') = \frac{1}{2}\sum_{a=0}^{3}
    ((\nabla_a\otimes\nabla_a)G_0)(\gamma(s),\gamma(s'))  
\end{equation} 
is the vacuum expectation value of the point-split energy density for the massless scalar field (see~(3.11) in~\cite{FewsterThompson2023}). Here, it is assumed that $\gamma$ is parametrised by proper time, while $\nabla_a=e_a^\mu\nabla_\mu$, where $e_a$ is a tetrad defined along $\gamma$ so that $e_0 = \dot{\gamma}$, and ~\eqref{eq:pointsplit_T_scalar} is defined as a distributional pull-back. 

In fact the QEI established in~\cite{GenQEI2000CFewster} also applies in general globally hyperbolic spacetimes, with normal ordering relative to any reference Hadamard state, 
whose Wightman function then replaces $G_0$ in~\eqref{eq:pointsplit_T_scalar}. However, only the Minkowski spacetime version will concern us here.

\subsection{Electromagnetic field}

The electromagnetic field in $4$-dimensional Minkowski spacetime has Lagrangian $\mathcal{L}=-\frac{1}{4}F_{\mu\nu}F^{\mu\nu}$ and classical stress-energy tensor
\begin{equation}\label{Classical_stress-tensor}
    T_{\mu\nu}=F_{\mu\alpha}F^{\alpha}_{\phantom{\alpha} \nu}+\frac{1}{4}\eta_{\mu\nu}F_{\alpha\beta}F^{\alpha\beta},
\end{equation}
where, as usual, $F_{\mu\nu}=\nabla_\mu A_\nu-\nabla_\nu A_\mu$ expresses the electromagnetic field strength $F_{\mu\nu}$ in terms of the electromagnetic potential $A_\mu$.

The electromagnetic field also obeys the weak energy condition. 
Indeed, if $u^\mu$ is a timelike unit vector at some spacetime point $x$, we again choose a tetrad $e^\mu_a$ at $x$ with $e_0^\mu=u^\mu$. Then the energy density seen by an observer at $x$ with $4$-velocity $u^\mu$ may be written
\begin{equation}\label{classical_energy_tetrad}
    T_{\mu\nu}u^\mu u^\nu=\frac{1}{4}\sum_{a,b=0}^{3}
    (e_{a}^{\mu}e_{b}^{\nu}F_{\mu\nu})^2, 
\end{equation} 
which is evidently nonnegative. 

Passing to quantum field theory and adopting the Feynman gauge, the electromagnetic potential may be represented on an indefinite inner product space with Minkowski vacuum two-point function
\begin{equation}\label{eq:A_2-point}
\langle A_{\mu}(x) A_{\nu}(x')\rangle = -\eta_{\mu\nu} G_0(x,x'),
\end{equation}
where $G_0$ is as in~\eqref{scalar_2-point} and the minus sign in~\eqref{eq:A_2-point} is a consequence of our signature convention. 
The field-strength tensor is represented on the physical Hilbert space and has vacuum two-point function
\begin{align}\label{eq:F_2-point}
    \langle F_{\mu\nu}(x) F_{\rho\sigma}(x')\rangle &= 
    -\eta_{\nu\sigma}((\nabla_\mu\otimes\nabla_\rho) G_0)(x,x')
    + \eta_{\mu\sigma}((\nabla_\nu\otimes\nabla_\rho) G_0)(x,x')\nonumber\\ &\qquad
    + \eta_{\nu\rho}((\nabla_\mu\otimes\nabla_\sigma) G_0)(x,x')
    - \eta_{\mu\rho}((\nabla_\nu\otimes\nabla_\sigma) G_0)(x,x').
\end{align}

Whilst the classical energy density \eqref{Classical_stress-tensor} is nonnegative, the expectation value of the quantized energy density may take negative values. However, the averaged energy density along a timelike curve is again constrained by QEIs.
Consider an observer following a timelike curve $\gamma$, parameterized by proper time, and choose any smooth tetrad $e_a^\mu$ defined around $\gamma$ so that $e_0|_{\gamma(s)}=\dot{\gamma}(s)$. It was shown in~\cite{FewsterPfenning:2003} that the normal ordered energy density along $\gamma$ satisfies the QEI~\eqref{eq:QEI} but with $T$ replaced by 
\begin{equation}\label{eq:pointsplit_T_em}
    T_{\textnormal{em}}(s,s') = \frac{1}{4}\sum_{a,b=0}^{3}
     \langle 
    F_{ab}(\gamma(s))F_{ab}(\gamma(s')) \rangle,
\end{equation} 
which is the point-split electromagnetic energy density in the vacuum state. Here, $F_{ab}=e^{\mu}_a e^\nu_bF_{\mu\nu}$ are the tetrad components of the field strength. In fact~\cite{FewsterPfenning:2003} established the above QEI in arbitrary globally hyperbolic spacetimes -- see the analogous comments for the scalar field above.

\subsection{Computing the point-split energy densities}

The point-split energy densities~\eqref{eq:pointsplit_T_scalar} and~\eqref{eq:pointsplit_T_em} may be computed as follows. First, the distributional derivatives of $G_0$ may be taken under the $\epsilon\to 0+$ limit to give
\begin{equation}
    (\nabla_{\mu} \otimes\mathbbm{1})G_0(x,x')=\lim_{\epsilon\rightarrow 0+}\frac{(x-x'-\ii\epsilon)_\mu}{2\pi^2\sigma_\epsilon(x,x')^2},
\end{equation}
and then
\begin{equation}
    (\nabla_{\mu} \otimes\nabla_{\nu})G_0(x,x')=\lim_{\epsilon\rightarrow 0+}\frac{-\eta_{\mu\nu}}{2\pi^2\sigma_\epsilon(x,x')^2}-\frac{2(x-x'-\ii\epsilon)_\mu (x-x'-\ii\epsilon)_\nu}{\pi^2\sigma_\epsilon(x,x')^3}.
\end{equation}
We can now pull back to the worldline, $\gamma(s)$, to find
\begin{equation}\label{eq:nabla_nabla_G}
    (\nabla_{\mu} \otimes\nabla_{\nu})G_0(\gamma(s),\gamma(s'))=\lim_{\epsilon\rightarrow 0+}\left(\frac{-\eta_{\mu\nu}}{2\pi^2\sigma_\epsilon(\gamma(s),\gamma(s'))^2}-\frac{2(\gamma(s)-\gamma(s'))_\mu (\gamma(s)-\gamma(s'))_\nu}{\pi^2\sigma_\epsilon(\gamma(s),\gamma(s'))^3}\right).
\end{equation}
In order to represent the contractions of the differentiated two-point function, we introduce matrix-valued functions 
$C,D:\RR\times\RR\to M_4(\RR)$ by
\begin{equation}\label{eq:Cab_def}
    C_{ab}(s,s') = \eta_{\mu\nu}e_a^{\mu}(s) e_b^\nu(s'),
\end{equation} 
which reduces to $C_{ab}(s,s)= \eta_{ab}$ at the diagonal, and
\begin{align}\label{Dab_definition}
    D_{ab}(s,s')=(\gamma(s)-\gamma(s'))_{\mu}e^{\mu}_{a}(s)(\gamma(s')-\gamma(s))_{\nu}e^{\nu}_{b}(s').
\end{align}  

Given these definitions, the vacuum point-split energy density of the massless scalar field is 
\begin{equation}\label{eq:T_formula_scalar}
    T(s,s') =   \lim_{\epsilon\rightarrow 0+}\left(\frac{-\Tr(C(s,s'))}{4\pi^2\sigma_\epsilon(\gamma(s),\gamma(s'))^2}+\frac{\Tr(D(s,s'))}{\pi^2\sigma_\epsilon(\gamma(s),\gamma(s'))^3}\right)
\end{equation} 
from~\eqref{eq:pointsplit_T_scalar}, where $\Tr$ denotes the matrix trace
$\Tr C=\sum_{a=0}^3 C_{aa}$.
Meanwhile, using~\eqref{eq:F_2-point}, the vacuum point-split energy density for electromagnetism may be written 
\begin{align}\label{eq:T_em}
    T_{\textnormal{em}}(s,s') &= -\frac{1}{2}\sum_{a=0}^{3}\sum_{b=0}^{3}
    \left( C_{bb}(s,s') ((\nabla_a\otimes\nabla_a)G_0)(\gamma(s),\gamma(s')) \right. \nonumber\\
    &\qquad\qquad\qquad -\left. C_{ab}(s,s') ((\nabla_b\otimes\nabla_a)G_0)(\gamma(s),\gamma(s'))\right), 
\end{align}
where (as above) we write $\nabla_{a}=e^a_\mu\nabla_\mu$, and consequently
\begin{align} 
    T_{\textnormal{em}}(s,s')&=\lim_{\epsilon\rightarrow 0+}\frac{1}{2}\sum_{a=0}^{3}\sum_{b=0}^{3}\left(
    \frac{C_{bb}C_{aa}-C_{ab}C_{ba}}{2\pi^{2}\sigma_\epsilon(\gamma(s),\gamma(s'))^{2}}-\frac{2(C_{bb}D_{aa}-C_{ab}D_{ba})}{\pi^{2}\sigma_\epsilon(\gamma(s),\gamma(s'))^{3}}\right) \nonumber \\
    &= \lim_{\epsilon\rightarrow 0+}\left(
    \frac{(\Tr C )^2-\Tr(C^2)}{4\pi^{2}\sigma_\epsilon(\gamma(s),\gamma(s'))^{2}}-\frac{(\Tr(C)\Tr(D)-\Tr(CD))}{\pi^{2}\sigma_\epsilon(\gamma(s),\gamma(s'))^{3}}\right).
\end{align}
on inserting~\eqref{eq:nabla_nabla_G} and suppressing the dependence of $C$ and $D$ on $s,s'$. There is now a simplification using the following result, whose proof is given in Appendix~\ref{appx:trace_squared_C}.
\begin{lemma}\label{lem:trace_squared_C}
   The identity $\Tr C(s,s')^2 = (\Tr C(s,s'))^2$ holds for any $s,s'\in\RR$.
\end{lemma}
It follows immediately that the electromagnetic point-split energy density is
\begin{align}\label{T_formula}
    T_{\textnormal{em}}(s,s') 
    &= -\frac{1}{\pi^2}\lim_{\epsilon\rightarrow 0+} 
    \frac{\Tr C(s,s')\Tr D(s,s')-\Tr(C(s,s')D(s,s'))}{\sigma_\epsilon(\gamma(s),\gamma(s'))^{3}} .
\end{align}
 
In principle,~\eqref{eq:T_formula_scalar} and~\eqref{T_formula} and can be used to give the QEIs for massless scalars or electromagnetism along any timelike trajectory in Minkowski spacetime. However, the calculation is generally only tractable for stationary worldlines, to which we now turn.

\subsection{Simplifications for stationary worldlines}

If $T(s,s')$ depends only on $s-s'$, i.e., $T(s+s'',s'+s'')=T(s,s')$ for all $s''$, we say that $T$ is translationally invariant. In this case, $T(s,s')=\Tc(s-s')$ for some distribution $\Tc\in\DD'(\RR)$ and the QEI may be simplified and slightly generalised (see~\cite{GenQEI2000CFewster,FewsterVerch:2003}) to 
\begin{equation}\label{eq:QEI_stationary}
    \int_{-\infty}^{\infty}\dd s\,|g(s)|^{2}\langle:T_{\mu\nu}\dot{\gamma}^{\mu}\dot{\gamma}^{\nu}:\rangle_{\omega}(\gamma(s))\geq-\int_{-\infty}^{\infty}\dd\alpha\,|\widehat{g}(\alpha)|^{2}Q(\alpha),
\end{equation}
which holds for all (not necessarily real-valued) $g\in C_0^\infty(\RR)$, where 
\begin{equation}
Q(\alpha)=\frac{1}{2\pi^{2}}\int_{-\infty}^{\alpha}\dd u\,\widehat{\mathcal T}(u).
\end{equation}
One may replace $Q$ by its even part $Q_\even$ in~\eqref{eq:QEI_stationary}, because $|\widehat{g}(\alpha)|^2$ is even for all $g\in C_0^\infty(\RR)$.

Even when the trajectory is stationary, the direct computation of $\widehat{\Tc}$ and $Q$ can be nontrivial. 
However, a further simplification was discovered in~\cite{FewsterThompson2023} for the massless scalar field. Namely, if -- as occurs for stationary worldlines --
$\Tc$ takes the general form
\begin{equation}\label{eq:Tc_general_form_scalar}
    \Tc(s) = \frac{3}{2\pi^2(s-\ii 0^+)^4}-\frac{A}{4\pi^2(s-\ii 0^+)^2}  + \Tc_\reg(s), 
\end{equation}
where $A\in \RR$ and $\Tc_\reg$ is a smooth, real and even function obeying $\Tc_\reg(s)=\mathcal{O}(s^{-2})$ as $|s|\to\infty$, then the even part of $Q$ is given by
\begin{equation}
    Q_\even (\alpha) = \frac{1}{32\pi^3}(\alpha^4+ 2A\alpha^2+B),
\end{equation}
with $B=16\pi^2 \Tc_\reg(0)$. The massless scalar QEI is therefore given by
\begin{equation}\label{eq:stationary_scalar_QEI}
     \int_{-\infty}^{\infty}\dd s\,\rvert g(s) \rvert ^{2} \langle :T_{\mu\nu}\dot{\gamma}^{\mu}\dot{\gamma}^{\nu}: \rangle _{\omega} (\gamma(s)) \geq -\frac{1}{16\pi^{2}}\int_{-\infty}^{\infty}\dd s\,(|g''(s)|^{2}+2A|g'(s)|^{2}+B|g(s)|^{2}).
\end{equation} 
This strategy was implemented 
for the Frenet--Serret tetrad along stationary worldlines in Section~3 of~\cite{FewsterThompson2023}. One of our main results is that one can improve on the QEI of~\cite{FewsterThompson2023} by using a tetrad that rotates at constant angular velocity relative to the Frenet--Serret tetrad.
 
The leading singularity for the point-split energy density of the electromagnetic field is twice that of the scalar field, owing to the existence of two independent polarisation states. By comparison with the scalar case, one readily sees that if
\begin{equation}\label{eq:Tc_general_form_em}
    \Tc_{\textnormal{e.m.}}(s)=
    2\left(  
    \frac{3}{2\pi^2(s-\ii 0^+)^4} 
    -  \frac{A_{\textnormal{e.m.}}}{4\pi^2(s-\ii 0^+)^2}\right)
    +\Tc_{\textnormal{e.m.},\reg}(s),
\end{equation}  
for some smooth, real and even function $\Tc_{\textnormal{e.m.},\reg}$ obeying $\Tc_\reg(s)=\mathcal{O}(s^{-2})$ as $|s|\to\infty$, then
\begin{equation}
    Q_{\textnormal{e.m.},\even} (\alpha) = \frac{1}{16\pi^3}(\alpha^4+ 2A_{\textnormal{e.m.}}\alpha^2+B_{\textnormal{e.m.}}),
\end{equation}
where 
\begin{equation}
    B_{\textnormal{e.m.}} = 8\pi^2 \Tc_{\textnormal{e.m.},\reg}(0).
\end{equation}
Comparing with~\eqref{eq:stationary_scalar_QEI}, the QEI for electromagnetism along stationary worldlines is 
\begin{equation}\label{eq:stationary_em_QEI}
    \int_{-\infty}^{\infty}\dd s\,\rvert g(s) \rvert ^{2} \langle :T_{\mu\nu}\dot{\gamma}^{\mu}\dot{\gamma}^{\nu}: \rangle _{\omega} (\gamma(s)) \geq -\frac{1}{8\pi^{2}}\int_{-\infty}^{\infty}\dd s\,(|g''(s)|^{2}+2A_{\textnormal{e.m.}}|g'(s)|^{2}+B_{\textnormal{e.m.}}|g(s)|^{2}).
\end{equation}

In Section~\ref{sec:QEIs}, we will show that the point-split energy densities $T$ and $T_\textnormal{e.m.}$ along stationary worldlines are translationally invariant for any frame rotating at constant angular velocity relative to the Frenet--Serret and also establish the decompositions~\eqref{eq:Tc_general_form_scalar} and~\eqref{eq:Tc_general_form_em} with remainders $\Tc_\reg$ and $\Tc_{\textnormal{e.m.},\reg}$ that are smooth, even real-valued functions of order
$\mathcal{O}(s^{-2})$ as $|s|\to\infty$ on the real line. Once that is done, we will optimise over the angular velocity to find the best possible QEI in this class.  
These various tasks require a more detailed discussion of stationary worldlines, to which we now turn.

\section{Stationary worldlines}\label{sec:stationary}

\subsection{Definition and classification of stationary worldlines}

Stationary worldlines may be defined in various ways -- see~\cite{FewsterThompson2023} and its bibliography. We adopt the following approach.
\begin{definition}
    A stationary worldline is
a smooth future-directed timelike curve $\gamma:\RR\to M$, parametrised by proper time, along which there exists a right-handed tetrad $e_a(s)$ with $e_0(s)=\dot{\gamma}(s)$, solving the Frenet--Serret equations 
\begin{equation} \label{eq:FrenetSerret}
    \dot{e}_{a}=K_{a}^{\ b}e_{b},  
\end{equation} 
where $K_a^{\phantom{a}b}$ is a constant tridiagonal real matrix  
\begin{equation}\label{eq:K}
K_\bullet^{\phantom{\bullet}\bullet}=
\begin{pmatrix}
    0 & \kappa & 0 & 0 \\ \kappa & 0 & \tau & 0 \\ 0 & -\tau & 0 & \upsilon\\
    0 & 0 & -\upsilon & 0
\end{pmatrix}
\end{equation}
with parameters $\kappa,\tau\in\RR^+$, $\upsilon\in\RR$, subject to the conditions $\kappa=0\implies \tau=0\implies \upsilon=0$, i.e., 
\begin{equation}\label{eq:ktu}
(\kappa,\tau,\upsilon)\in \left([0,\infty)\times\{0\}\times\{0\}\right)\cup \left((0,\infty)\times(0,\infty)\times\RR\right).
\end{equation}
The parameters $\kappa$, $\tau$ and $\upsilon$ are called the curvature, torsion and hypertorsion respectively. 
\end{definition}
Curvature, torsion and hypertorsion have dimensions of inverse time in our units. The simplest examples of stationary worldlines are inertial worldlines, with $\kappa=\tau=\upsilon=0$ and constant tetrad, while uniformly linearly accelerated worldlines correspond to $\kappa>\tau=\upsilon=0$ with $e_2$ and $e_3$ constant. A further four classes of stationary worldline will be described below. The matrix $K_\bullet^{\phantom{\bullet}\bullet}$ will usually be written simply as $K$; it is neither symmetric nor antisymmetric, but the matrix $K_{ab}=K_a^{\phantom{a}{c}}\eta_{cb}$ is antisymmetric.

A particular consequence of~\eqref{eq:FrenetSerret} is that $\dot{e}_0=\kappa e_1$, so because $e_1$ is a unit spacelike vector we have $\kappa^2=-\eta(\dot{e}_0,\dot{e}_0)=-\eta(\ddot{\gamma},\ddot{\gamma})$, which shows that $\kappa$ is the proper acceleration of $\gamma$. In fact the other curvatures are also determined by $\gamma$ together with~\eqref{eq:FrenetSerret} and~\eqref{eq:K}.
\begin{proposition}\label{prop:FS}
    Let $\gamma$ be a stationary worldline. Then $\gamma$ determines its curvature matrix $K$ uniquely. If $\kappa,\tau>0$ then the tetrad $e_a$ is uniquely determined; if $\kappa>\tau=\upsilon=0$ then $e_a$ is unique up to an arbitrary rigid $\SO(2)$ rotation of $e_2$ and $e_3$; if $\kappa=\tau=\upsilon=0$ then $e_a$ is unique up to arbitrary rigid $\SO(3)$ rotations of $e_1$, $e_2$ and $e_3$.
\end{proposition}
\begin{proof}
    It has already been noted that $\kappa$ is the proper acceleration.
    If $\kappa=0$ then also $\tau=\upsilon=0$, so $e_a$ is a constant right-handed tetrad with $e_0=\dot{\gamma}$, thus being fixed up to a rigid $\SO(3)$ rotation of the spacelike basis vectors.
    Otherwise, $\kappa>0$. It follows from~\eqref{eq:FrenetSerret} and~\eqref{eq:K} that
    $e_1=\kappa^{-1}\dot{e}_0$, whereupon, as $\dot{e}_1=\kappa e_0+\tau e_2$, $\tau\in[0,\infty)$ is determined by $\tau^2=-\eta(\dot{e}_1-\kappa e_0,\dot{e}_1-\kappa e_0)$. If $\tau=0$ then also $\upsilon=0$ and then $e_2$ and $e_3$ are constant unit vectors orthogonal to each other and $e_0$, $e_1$. Hence they are fixed up to a rigid $\SO(2)$ rotation. Otherwise we now have $\kappa>0,\tau>0$, and it follows from~\eqref{eq:FrenetSerret} and~\eqref{eq:K} that  $e_2=\tau^{-1}(\dot{e}_1-\kappa e_0)$. The vector $e_3$ is then uniquely fixed so that $e_a$ is a right-handed tetrad, and $\upsilon\in\RR$ is uniquely determined by~\eqref{eq:FrenetSerret} and~\eqref{eq:K} so that $\upsilon e_3 = \dot{e}_2+\tau e_1$.
\end{proof}
If $\gamma$ is a stationary worldline, any solution to its corresponding Frenet--Serret equations will be called a \emph{Frenet--Serret tetrad} for $\gamma$ and denoted $e^\FS$. Proposition~\ref{prop:FS} shows that a stationary worldline determines its Frenet--Serret tetrad uniquely up to a (possibly trivial) subgroup of $\SO(3)$ acting rigidly on the spacelike tetrad vectors. 

Conversely, suppose a matrix $K$ is given as in~\eqref{eq:K} with $(\kappa,\tau,\upsilon)$ obeying~\eqref{eq:ktu}. Given any right-handed tetrad $e_a$ with future-pointing timelike $e_0$, the unique solution $e^\FS_a(s)$ to the Frenet--Serret equation~\eqref{eq:FrenetSerret} with $e_a^\FS(0)=e_a$ is
\begin{equation}\label{eq:FS_sol}
    e_a^\FS(s) = \exp(sK)_{a}^{\phantom{a}b}e_{b}.
\end{equation}
Here, the matrix exponential may be defined using its power series, as will be explained in Section~\ref{sec:functional_calc}. Then through any point $p$ in Minkowski spacetime, there is a unique
stationary worldline $\gamma$ with $\gamma(0)=p$ and for which $e^\FS$ is a Frenet--Serret tetrad, obtained by integrating the equation $\dot{\gamma}(s)=e_0^\FS(s)$. Using the identity
\begin{equation}
    \frac{\dd}{\dd s} (sF(sK)) = F(sK) + sK F'(sK) = \left.\frac{\dd}{\dd z}(z F(z))\right|_{z=sK},
\end{equation}
which holds for any entire analytic $F$, the unique solution is 
\begin{equation}\label{eq:gammas}
    \gamma^\mu(s) = p^\mu + s\expc(sK)_0^{\phantom{0}b} e^\mu_b,
\end{equation}
where the cardinal exponential function $\expc$ was defined in~\eqref{eq:sin(h)c}.\footnote{If $K$ is invertible, we may also write
$\gamma^\mu(s)=p^\mu+ [(\exp(sK)-I)K^{-1}]_0^{\phantom{0}b}e_b^\mu$ -- however $K$ is only invertible for $\upsilon>0$, whereas~\eqref{eq:gammas} is valid for all parameter values~\eqref{eq:ktu}.} As $\gamma$ was determined by 
$K$, $p$ and $e_a$, it follows that all stationary worldlines with given $K$ are related by proper orthochronous Poincar\'e transformations and reparametrisations $s\mapsto s+\textnormal{const}$.

The stationary worldlines fall into 6 types, determined by the values of $\kappa$, $\tau$ and $\upsilon$ and the eigenvalue structure of $K$, which is established in Appendix~\ref{appx:proof_of_lems}.
\begin{lemma}\label{lem:K1}
    Let $\kappa,\tau,\upsilon$ obey~\eqref{eq:ktu}. Unless 
    $\kappa=\tau>\upsilon=0$, the matrix $K$ is diagonalisable 
    with eigenvalues $\pm\chi$ and $\pm\ii\omega$, where $\chi\ge 0$ and
    $\omega\ge 0$ are given explicitly by     
    \begin{align}
         \chi^2 & = \tfrac{1}{2}(\kappa^2-\tau^2-\upsilon^2 + R) \label{eq:chi}\\
         \omega^2 & = \tfrac{1}{2}(-\kappa^2+\tau^2+\upsilon^2 + R), 
         \label{eq:omega}
     \end{align}
     where 
      \begin{equation}\label{eq:R}
         R(\kappa,\tau,\upsilon) = \sqrt{(\kappa^2-\tau^2-\upsilon^2)^2 + 4\kappa^2\upsilon^2}\ge 0.
     \end{equation}
     If $\kappa=\tau>\upsilon=0$ then $K$ is non-diagonalisable and its only eigenvalue is $0$ with geometric multiplicity $2$.
\end{lemma} 
    When simplifying expressions involving $\chi$ and $\omega$, it is often useful to use the identities
    \begin{equation}\label{eq:4thpowers}
        \chi^4 = (\kappa^2-\tau^2-\upsilon^2)\chi^2+\kappa^2\upsilon^2 
        \qquad \omega^4 = -(\kappa^2-\tau^2-\upsilon^2)\omega^2+\kappa^2\upsilon^2,        
    \end{equation}
    which hold because $\chi$ and $\ii\omega$ are roots of the characteristic polynomial of $K$ (see~\eqref{eq:characteristic}).

Lemma~\ref{lem:K} suggests a natural partition of the parameter space for stationary worldlines into six types, as set out in
Table~\ref{tab:inv}. Physically, these curves
represent inertial, uniformly linearly accelerated, uniformly rotating (elliptic), parabolic, catenary or loxodromic motions.  The three cases with $\chi>0$, namely the linearly accelerated, catenary, and loxodromic motions, can be grouped together under the term `hyperbolic motions'. 
We refer to~\cite{FewsterThompson2023} for further details on each of these motions and~\cite{Bunney:2024} for a classification of stationary worldlines in Minkowski spacetime of general dimension.
 
\begin{table}
    \centering
    \begin{tabular}{|c|c|c|c|c|c|}\hline
         Inertial & Linear Acc. &  Catenary &  Parabolic & Elliptic  & Loxodromic  \\ \hline 
         $\kappa=\tau=\upsilon=0$ & $\kappa>0$ & $\kappa>\tau>0$ & $\kappa=\tau>0$ & $\tau>\kappa>0$ & $\kappa,\tau>0$\\
          &  $\tau=\upsilon=0$ & $\upsilon=0$ & $\upsilon=0$ & $\upsilon=0$ & $\upsilon\neq 0$
         \\ \hline
         $\chi=\omega=0$ & $\chi=\kappa$  &
         $\chi=\sqrt{\kappa^2-\tau^2}$ & Nondiag. & $\chi=0$ & $\chi,\omega>0$\\
         & $\omega=0$ & $\omega=0$ &  &  $\omega=\sqrt{\tau^2-\kappa^2}$ & \\
         \hline
    \end{tabular}
    \caption{Parameter ranges corresponding to the six types of stationary worldline and the corresponding values of $\chi$ and $\omega$ in~\eqref{eq:chi},~\eqref{eq:omega}. In all cases other than the parabolic, $K$ is diagonalisable. The linearly accelerated, catenary and loxodromic motions are collectively described as hyperbolic motions.} 
    \label{tab:inv}
\end{table}

\subsection{Relatively rotating frames}

We will study QEIs along stationary worldlines, using frames that rotate relative to the Frenet--Serret tetrad. Let $\gamma$ be a stationary worldline and suppose that $(e_a(s))_{a=0}^3$ is any smooth right-handed frame along $\gamma$ with $e_0=\dot{\gamma}$. Then there is a smooth map $s\mapsto R(s)\in \{1\}\times \SO(3)\subset \SO(1,3)$ so that
\begin{equation}
    e_a(s) = R(s)_a^{\phantom{a}b}e^\FS_b(s), 
\end{equation}
and the block diagonal structure of $R(s)$ entails 
\begin{equation}
    R(s) \eta = \eta R(s), \qquad R(s)^{-1} = R(s)^T.
\end{equation}
For reasons explained below, we will eventually focus on the case
\begin{equation}\label{eq:RandOmega}
    R(s) = R(0)\exp(-s\Omega) \qquad \Omega = \begin{pmatrix}
        0 & 0 & 0 & 0\\ 0 & 0 & \varpi_3 & -\varpi_2 \\  0 & -\varpi_3 & 0 & \varpi_1 \\ 0 & \varpi_2 & -\varpi_1 & 0
    \end{pmatrix},
\end{equation}
in which the frame $(e_a)_{a=0}^3$ has constant angular velocity $-\boldsymbol{\varpi}$ relative to the Frenet--Serret frame. Note that $\Omega$ is antisymmetric and commutes with $\eta$.

Next, recall that the Fermi--Walker derivative along $\gamma$ is defined by 
\begin{equation}
    \frac{D^\FW}{\dd s}V = \nabla_u V + \eta(a, V) u - \eta(u, V) a,
\end{equation}
where $u=\dot{\gamma}=e_0^\FS$ and $a=\ddot{\gamma}=\kappa e_1^\FS$ are the $4$-velocity and $4$-acceleration of $\gamma$. A calculation gives
\begin{equation}
     \frac{D^\FW}{\dd s} e_a = [(\dot{R}+R\Omega^\FW)R^T]_a^{\phantom{a}b}e_b
\end{equation}
where $\Omega^\FW$ corresponds to $\boldsymbol{\varpi}=\boldsymbol{\varpi}^\FW:=(\upsilon,0,\tau)$. The tetrad $(e_a)_{a=0}^3$ is a \emph{Fermi--Walker tetrad} if it is invariant under Fermi--Walker transport, which happens precisely for $R(s)=R(0)\exp(-s \Omega^\FW)$.

\subsection{Functional calculus}\label{sec:functional_calc}

The matrices $C(s,s')$ and $D(s,s')$ used to construct the point-split energy density may be expressed in terms of analytic functions applied to $sK$, where $s$ has dimensions of time making $sK$ dimensionless. We now set out the construction and properties of such expressions, taking account of the fact that $K$ is not hermitian and can fail to be diagonalisable.
To begin, we assemble various useful facts about the matrix $K$, which are proved in Appendix~\ref{appx:proof_of_lems}. Similar results appear in~\cite{Synge1967}.
\begin{lemma}\label{lem:K}
(a) Suppose $K$ is diagonalisable (i.e., it is not the case that $\kappa=\tau>\upsilon=0$), with 
    \begin{equation}\label{eq:K_diagonal_form}
         K = P \Lambda P^{-1}, \qquad \Lambda = \begin{pmatrix}
             \chi & & & \\ & -\chi &   &  \\ && \ii \omega & \\ &&& -\ii\omega
         \end{pmatrix}  
     \end{equation}
     for $\chi,\omega\ge 0$ as in Lemma~\ref{lem:K1} and some invertible matrix $P\in M_4(\CC)$.  
    Then $K^2$ has eigenvalues $\chi^2$ and $-\omega^2$. If $K\neq 0$ there are (generally non-orthogonal) rank-$2$ projections $\Pi$ and $I-\Pi$ onto the corresponding eigenspaces given by\footnote{See~\eqref{eq:Piexplicit} and~\eqref{eq:KPi} for explicit matrix representations.}  
        \begin{equation}\label{eq:Pi}
             \Pi = P\diag(1,1,0,0)P^{-1}=\frac{\omega^2 I + K^2}{\chi^2+\omega^2}, \quad
             I-\Pi = P\diag(0,0,1,1)P^{-1}=\frac{\chi^2 I - K^2}{\chi^2+\omega^2}.
         \end{equation}
    If $\chi>0$ then $\Pi_+=P\diag(1,0,0,0)P^{-1}$ and $\Pi_-=P\diag(0,1,0,0)P^{-1}$ are rank-$1$ projections onto the eigenspaces of $K$ with eigenvalue $\chi$ and $-\chi$ respectively given by
     \begin{equation}\label{eq:Pi_pm}
        \Pi_\pm = \frac{1}{2\chi}(\chi I \pm K)\Pi.
    \end{equation} 
  
     (b) If $K$ is nondiagonalisable, i.e., $\kappa=\tau>0$, $\upsilon=0$, then $K^r=0$ for $r\ge 3$. Furthermore $\Tr K^r=0$ for all $r\ge 1$.  
     \end{lemma}    
      
     If $F$ is an entire analytic function, $F(s K)$ may be defined as the absolutely convergent series
     \begin{equation}\label{eq:FofK}
         F(sK) = \sum_{n=0}^\infty \frac{F^{(n)}(0)}{n!} (sK)^n,
     \end{equation}
     whereupon $s\mapsto F(sK)$ is entire analytic, 
     with derivative $F'(sK)K$. One has the following general properties, proved in Appendix~\ref{appx:proof_of_lems}.
     \begin{lemma}\label{lem:FofK} 
        (a) If, using the notation of Lemma~\ref{lem:K}, $K$ is diagonalisable and at least one of the parameters $\chi$, $\omega$ is nonzero, then
         \begin{equation}\label{eq:FofK_diag}
             F(sK) = (F_{\textnormal{e}}(\chi s)I + s F_{\textnormal{oc}}(\chi s)K))\Pi + 
             (F_{\textnormal{e}}(\ii\omega s)I + s F_{\textnormal{oc}}(\ii\omega s)K))(I-\Pi) ,
         \end{equation}
         where $F_\textnormal{e}$ is the even part of $F$ and $F_{\textnormal{oc}}$ is the cardinal function of its odd part,
         \begin{equation}
         F_{\textnormal{e}}(z)=\frac{1}{2}(F(z)+F(-z)), \qquad  F_{\textnormal{oc}}(z)=\frac{1}{2z}(F(z)-F(-z))
         \end{equation}
         (with the understanding that $F_\textnormal{oc}(0)=F'(0)$).

         The trace of $F(sK)$ is given by 
         \begin{equation}\label{eq:TrFK}
             \Tr F(sK) = 2 F_\textnormal{e}(\chi s)  +  2 F_\textnormal{e}(\ii\omega s).
         \end{equation}
        If $G$ is also entire analytic then
        \begin{equation}
            (FG)(sK) = F(sK)G(sK)= G(sK)F(sK).
        \end{equation}
         
         (b) If $K$ is diagonalisable and $\chi=\omega=0$ then $K=0$ and consequently $F(sK)=F(0)I$ for any analytic $F$, giving 
         $\Tr F(sK)=4F(0)$.

         (c) If $K$ is not diagonalisable then $K^3=0$ and consequently
         \begin{equation}\label{eq:FofK_nondiagonalisable}
             F(sK) = F(0)I + F'(0)sK + \frac{1}{2}F''(0)s^2K^2.
         \end{equation}
         One has $\Tr F(sK)=4F(0)$.

         (d) In all cases, one has 
         \begin{equation}\label{eq:functional_calc_id}
            \eta F(sK)^T  = F(-sK)\eta,
         \end{equation}
        where $\eta=\diag(1,-1,-1,-1)$.
     \end{lemma}
     
     In fact, all but part~(c) of Lemma~\ref{lem:FofK} also apply to any continuous $F$, but part~(c) falls outside the continuous functional calculus -- for simplicity we have stated everything for the analytic case. Note also that formula~\eqref{eq:FofK_nondiagonalisable}
     applies in both parts~(b) and~(c), and we will sometimes separate computations of $F(sK)$ into the two cases depending on whether (i) $K$ is diagonalisable with $\chi^2+\omega^2>0$, or (ii) $K$ is nondiagaonalisable or zero.

    We illustrate Lemma~\ref{lem:FofK} with three relevant examples, namely $\exp(sK)$, $\expc(sK)$ and $\sinhc^2(sK)$, where $\expc$ and $\sinhc$ were defined in~\eqref{eq:sin(h)c}. The first two appear in the expression giving $\gamma(s)$ in terms of its initial conditions~\eqref{eq:gamma_sol}, while the third is used in the calculation of $\sigma(\gamma(s),\gamma(s'))$ in Section~\ref{sec:summary} 

    First, for $F(z)=\exp(sz)$, Lemma~\ref{lem:FofK} gives 
    \begin{equation}\label{eq:expsK}
        \exp(sK) =  (\cosh(\chi s) I + s \sinhc(\chi s) K  ) \Pi
        +  (\cos(\omega s) I + s\sinc(\omega s) K)(I-\Pi),
    \end{equation}
    for the diagonalisable case with $\chi^2+\omega^2>0$, and alternatively
    \begin{equation}\label{eq:expsK_nondiag}
        \exp(sK) = I + s K + \frac{s^2}{2}K^2
    \end{equation}
    if either $K=0$ (i.e., $\exp(sK)=I$) or $K$ is nondiagonalisable. These formulae can be checked directly using the differential equation for the left-hand side. 

    Second, starting from the definition of $\expc$, we find
    \begin{equation}
        \expc_{\textnormal{e}}(z) = \sinhc(z),\qquad \expc_{\textnormal{oc}}(z) = 
        \frac{\cosh(z)-1}{z^2}=\frac{1}{2}\sinhc^2(z/2),
    \end{equation}
    whereupon Lemma~\ref{lem:FofK} gives
    \begin{equation}\label{eq:fsK}
        \expc(sK) = \left(\sinhc(\chi s) I + \frac{s}{2}\sinhc^2(\chi s/2)K\right)\Pi + 
        \left(\sinc(\omega s) I + \frac{s}{2}\sinc^2(\omega s/2)K\right)(I-\Pi)
    \end{equation}
    if $K$ is diagonalisable with $\chi^2+\omega^2>0$ or alternatively
    \begin{equation}
        \expc(sK) = I + \frac{s}{2}K + \frac{s^2}{6}K^2
    \end{equation}
    in the case where $K$ is nondiagonalisable or zero.

    Third, for the even function $F(z)=\sinhc^2(sK)$, we find 
    \begin{equation}
        \sinhc^2(sK) = \sinhc^2(\chi s) \Pi + \sinc^2 (\omega s) (I-\Pi)
    \end{equation} 
    if $K$ is diagonalisable and $\chi^2+\omega^2>0$, or alternatively
    \begin{equation}
        \sinhc^2(sK) = I + \frac{s^2}{3}K^2
    \end{equation}
    if $K$ is nondiagonalisable or zero.

    \subsection{Properties 
    of $\sigma$, $\Tr C$, $\Tr D$, and $\Tr CD$ for stationary worldlines}\label{sec:summary}

    Using the functional calculus of Lemma~\ref{lem:FofK} we can compute various quantities needed to evaluate the point-split energy densities for the scalar field and electromagnetism. To keep the discussion focussed, we will give the details for the geodesic separation and $\Tr C$, indicating the general method, and relegate the other parts to Appendix~\ref{appx:CD}. We will write
    $\exp(sK)$ as $\ee^{sK}$, and similarly for other matrix exponentials, to reduce confusion between $\exp$ and $\expc$.

Starting with the geodesic separation, we may integrate the velocity between 
proper times $s'$ and $s$ using the same method as in~\eqref{eq:gammas}, to obtain
\begin{align}\label{eq:gamma_sol}
\gamma(s)- \gamma(s')&=
 (s-s')\expc((s-s')K)_0^{\phantom{0}b}e_b^\FS(s')\nonumber\\
 &= (s-s') [\expc((s-s')K)\ee^{s'K}]_{0}^{\phantom{0}b}e^\FS_b(0).
\end{align} 
Thus, the signed squared geodesic separation between $\gamma(s)$ and $\gamma(s')$ is
\begin{align}
    \sigma(\gamma(s),\gamma(s')) &= -\eta(\gamma(s)-\gamma(s'),\gamma(s)-\gamma(s')) \nonumber\\
    &= 
    - (s-s')^2 [\expc((s-s')K)\ee^{s'K}]_{0}^{\phantom{0}a}\eta_{ab}
    [\expc((s-s')K)\ee^{s'K}]_{0}^{\phantom{0}b} \nonumber\\
    &=
    -(s-s')^2[\expc((s-s')K)\ee^{s'K} \eta (\expc((s-s')K)\ee^{s'K})^T]_{00} \nonumber\\
    &=-(s-s')^2[\expc((s-s')K) \expc(-(s-s')K)\eta]_{00} \nonumber \\
    &=-(s-s')^2[\sinhc^2(\tfrac{1}{2}(s-s')K)]_{00},
\end{align}
using~\eqref{eq:gamma_sol} together with~\eqref{eq:functional_calc_id}, and the facts that any two functions of $sK$ commute,
that $\eta$ is diagonal with $\eta_{00}=1$ and that $\expc(z)\expc(-z) = \sinhc^2(z/2)$. 
Thus
\begin{equation}\label{eq:sig_Upsilon}
    \sigma(\gamma(s),\gamma(s')) = -(s-s')^2\Upsilon(\kappa^2(s-s')^2), 
\end{equation} 
where 
\begin{equation}\label{eq:Upsilon_def}
\Upsilon(z) = [\sinhc^2(\tfrac{1}{2}\sqrt{z} K/\kappa)]_{00}
\end{equation}
is an entire analytic function, noting that there are no branch cut issues to address because $\sinhc$ is even.

Explicit expressions for $\Upsilon$ may be obtained by applying the functional calculus to~\eqref{eq:Upsilon_def}. In the inertial ($K=0$) and parabolic (nondiagonalisable) cases, Lemma~\ref{lem:FofK}(b,c) give
\begin{equation}\label{eq:Upsilon_exact_b}
    \Upsilon(z) = \begin{cases} 1 & \text{inertial} \\ 1 + \frac{1}{12}z & \text{parabolic}
    \end{cases}
\end{equation}
by considering the Taylor expansion $\sinhc^2 (\tfrac{1}{2}w\sqrt{z}/\kappa) = 1+\tfrac{z}{12\kappa^2}w^2 +\mathcal{O}(w^4)$ and using $[K^2]_{00}=\kappa^2$. In the other four cases, $K$ is diagonalisable and $\chi^2+\omega^2>0$, and then Lemma~\ref{lem:FofK}(a) gives 
\begin{equation}\label{eq:Upsilon_exact_a}
    \Upsilon(z) = \frac{\kappa^2+\omega^2}{\chi^2+\omega^2} \sinhc^2 \frac{\chi\sqrt{z}}{2\kappa} + \frac{\chi^2-\kappa^2}{\chi^2+\omega^2}\sinc^2 \frac{\omega\sqrt{z}}{2\kappa} ,
\end{equation}
again noting that there are no branch cut issues because $\sinhc$ and $\sinc$ are even.

In all cases, we see that $\Upsilon$ is nonvanishing and real-valued on the positive real axis (as one could also see by noting that no distinct points on a timelike trajectory in Minkowski spacetime can be null related) and so is bounded away from zero in a neighbourhood of $\RR^+$ in $\CC$.
The fact that~\eqref{eq:sig_Upsilon} holds for an entire analytic $\Upsilon$ was proved in Appendix A of~\cite{FewsterThompson2023}, though that reference only gave the initial terms of the Taylor series (for $\kappa>0)$
\begin{equation}\label{trace_geo_series}
    \Upsilon(z)=1+\frac{1}{12}z+\frac{\kappa^2-\tau^2}{360\kappa^2}z^2+O(z^3).
\end{equation} 
Here, by contrast, we have the exact formulae~\eqref{eq:Upsilon_exact_a} and~\eqref{eq:Upsilon_exact_b}.
 
The matrix $C(s,s')$ depends on the particular choice of tetrad chosen along $\gamma$. Its trace has the following properties.
\begin{lemma}\label{lem:C}
    For $R(s)=R(0)\ee^{-s\Omega}$ as in~\eqref{eq:RandOmega}, we have
    \begin{equation}\label{eq:TrC_simp}
    \Tr C(s,s') = \Tr (\ee^{-(s-s')\Omega}\ee^{(s-s')K}\eta) ,
\end{equation}
    so $\Tr C(s,s')$ is translationally invariant. If $\kappa>0$ then 
\begin{equation}\label{eq:Gdef}
\Tr C(s,s') = \Gc_\Omega(\kappa^2(s-s')^2)
\end{equation}
for all $s,s'\in\RR$, where $\Gc_\Omega$ is entire analytic, real-valued on $\RR^+$, and has Taylor expansion
\begin{equation}
    \Gc_\Omega(z) = -2 + \frac{z}{\kappa^2}\|\boldsymbol{\varpi}-\boldsymbol{\varpi}^\FW\|^2+ 
    \frac{z^2}{12\kappa^4} G_\Omega
    + O(z^3),
\end{equation}
with
\begin{align}
    G_\Omega&= \kappa^{2}\left(3\varpi_2^2+ \left(\varpi_{3}-\tau \right) \left(3 \varpi_{3}-\tau \right)\right)
    -\|\boldsymbol{\varpi}-\boldsymbol{\varpi}^\FW\|^4
     -\|\boldsymbol{\varpi}\times\boldsymbol{\varpi}^\FW\|^2.
\end{align}
If $\kappa=0$, we have 
\begin{equation}\label{eq:TrC_kappa0}
\Tr C(s,s')=\Tr (\ee^{-(s-s')\Omega} \eta) = -2\cos (\|\boldsymbol{\varpi}\|(s-s')).
\end{equation}
\end{lemma}
The factor of $\kappa$ in~\eqref{eq:Gdef} is inserted for dimensional reasons and also for consistency with similar expressions in~\cite{FewsterThompson2023}; indeed, $\Gc_{\Omega}|_{\Omega=0}$ corresponds to what was written as $\sum_{a=0}^3 G_a$ in~\cite{FewsterThompson2023}. In that case, $G_{\Omega}|_{\Omega=0}=\kappa^2\tau^2-(\tau^2+\upsilon^2)^2$ and the expansion agrees with Eq.~(4.9) in~\cite{FewsterThompson2023} (except that the $12$ in the denominator of the $O(z^2)$ term in that reference is absent due to a typographical error).
A second interesting case is where $\boldsymbol{\varpi}=\boldsymbol{\varpi}^\FW$, and the Taylor series of $\Gc_\Omega$ simplifies to 
\begin{equation}
\Gc^\FW(z) = -2 
    + O(z^3).
\end{equation}
However, $\boldsymbol{\varpi}=\boldsymbol{\varpi}^\FW$ is not a critical point of the $\mathcal{O}(z^2)$ term owing to the term $(\varpi_{3}-\tau ) (3 \varpi_{3}-\tau)$. 
\begin{proof}[Proof of Lemma~\ref{lem:C}.]
In the first instance, suppose that $(e_a)_{a=0}^3$ is a general smooth right-handed tetrad with $e_0=\dot{\gamma}$. As described above, it is related to the Frenet--Serret tetrad at $s=0$ by 
\begin{equation}
e_a(s) = R(s)_{a}^{\phantom{a}b}e_b^\FS(s) = [R(s)\ee^{sK}]_a^{\phantom{a}b}e_b^\FS(0)
\end{equation}
for some smoothly varying rotation $R(s)\in\SO(1,3)$. 

Starting from~\eqref{eq:Cab_def} and~\eqref{eq:FS_sol}, one has
\begin{align}
    C_{ab}(s,s') &= \eta(e_a(s),e_b(s')) = 
     [R(s)\ee^{sK}]_a^{\phantom{a}c}[R(s')\ee^{s'K}]_b^{\phantom{b}d}\eta(e_c^\FS(0),e_d^\FS(0))  \nonumber\\    
&  = [R(s)\ee^{sK}]_a^{\phantom{a}c}\eta_{cd}[R(s')\ee^{s'K}]_b^{\phantom{b}d},
\end{align}
so, in matrix form, 
\begin{align}\label{eq:Cssp}
C(s,s') &= R(s)\ee^{sK} \eta (\ee^{s'K})^T R(s')^T  = R(s)\ee^{(s-s')K}\eta R(s')^{-1} \nonumber \\ &= R(s) C^\FS(s,s')R(s')^{-1},
\end{align}
using~\eqref{eq:functional_calc_id}, where $C^\FS(s,s')=\ee^{(s-s')K}\eta$ is the matrix obtained using the Frenet--Serret tetrad. Clearly, $C^\FS$ is translationally invariant, but the same is not true for general matrices $C$.
However, if we specialise to the situation where $R(s) = R(0)\ee^{-s\Omega}$ as in~\eqref{eq:RandOmega}, we may rotate the order of factors under the trace to obtain~\eqref{eq:TrC_simp}, which is evidently translationally invariant and independent of $R(0)$. This explains why we study tetrads rotating uniformly with respect to the Frenet--Serret tetrad.

Next, note that 
\begin{align}
    \ee^{s\Omega}\ee^{-sK}\eta &= \ee^{s\Omega}\eta (\ee^{sK})^T =  (\ee^{-s\Omega})^T \eta (\ee^{sK})^T  =(\ee^{sK}\eta\ee^{-s\Omega})^T
\end{align}
and hence $\Tr C(s,s')$ is even in $s-s'$ because
\begin{equation}
    \Tr (\ee^{s\Omega}\ee^{-sK}\eta) = \Tr (\ee^{sK}\eta\ee^{-s\Omega}) = \Tr (\ee^{-s\Omega}\ee^{sK}\eta).
\end{equation}
Both matrix functions appearing in~\eqref{eq:TrC_simp} are analytic in $s-s'$. Hence, for $\kappa>0$, there is an entire analytic function $\Gc_\Omega$ satisfying~\eqref{eq:Gdef}, and which must be real-valued on $\RR^+$ because $C(s,s')$ is a real matrix. The series expansion of $\Gc_\Omega$ may be computed in symbolic algebra using~\eqref{eq:expsK} or~\eqref{eq:expsK_nondiag} to evaluate $\ee^{sK}$ as appropriate, and then simplifying to find the stated expressions. The above arguments are independent of whether $K$ is diagonalisable.
 
Finally, if $\kappa=0$, which is equivalent to $K=0$, 
we have $\Tr C(s,s')=\Tr (\ee^{-(s-s')\Omega} \eta)$. The final equality in~\eqref{eq:TrC_kappa0} is standard and
can be seen easily by noting that the $3\times 3$ rotation matrix in the lower right of $\ee^{-s\Omega}$ has eigenvalues $1$ and $\ee^{\pm\ii\|\boldsymbol{\varpi}\|s}$. Multiplication by $\eta$ reverses the sign of these eigenvalues and on summing and adding $1$ for the top left component, the stated formula follows.
\end{proof}

\begin{lemma}\label{lem:D}
For $R(s)=R(0)\ee^{-s\Omega}$ as in~\eqref{eq:RandOmega}, we have
\begin{equation}\label{eq:TrDssp}
    \Tr D(s,s') = -(s-s')^2 [\expc((s-s')K)\eta \ee^{-(s-s')\Omega} \expc((s-s')K)]_{00},
\end{equation}
so $\Tr D(s,s')$ is translationally invariant. If $\kappa>0$, then
\begin{equation}\label{eq:Hdef}
    \Tr D(s,s') = -(s-s')^2 \Hc_\Omega(\kappa^2(s-s')^2)
\end{equation}
for all $s,s'\in\RR$, where $\Hc_\Omega$ is entire analytic, real-valued on $\RR^+$, and has Taylor expansion
\begin{equation}\label{eq:H_Taylor}
    \Hc_\Omega(z) = 1+\frac{z}{12} + 
   \frac{z^2}{360\kappa^2}\left(\kappa^2+4\tau^2+45\varpi_2^2 +15(\varpi_3-\tau)(3\varpi_3-\tau)\right)
    +O(z^3).
\end{equation}
If $\kappa=0$, then
\begin{equation}\label{eq:TrD_kappa0}
    \Tr D(s,s') =-(s-s')^2.
\end{equation}
\end{lemma}
Note that the dependence on $\Omega$ only appears at $O(z^2)$, as can be seen directly because $[K\eta\Omega]_{00}$,  
$[\eta\Omega K]_{00}$ and $[\Omega^2]_{00}$ each vanish.
If $\boldsymbol{\varpi}=0$ we obtain agreement with (4.10) of~\cite{FewsterThompson2023}. Meanwhile, if $\boldsymbol{\varpi}=\boldsymbol{\varpi}^\FW$, we find
\begin{equation}
 \Hc^\FW(z)   =1+\frac{z}{12}+\frac{\left(\kappa^{2}+4 \tau^{2}\right) z^{2}}{360 \kappa^{2}}+\mathcal{O}(z^3),
\end{equation}
but $\boldsymbol{\varpi}=\boldsymbol{\varpi}^\FW$ is not a critical point of the $\mathcal{O}(z^2)$ term owing to the term $(\varpi_{3}-\tau ) (3 \varpi_{3}-\tau)$.

Lemma~\ref{lem:D} is proved in Appendix~\ref{appx:CD}, along with the following result
on $\Tr CD$.
\begin{lemma}\label{lem:CD}
For $R(s)=R(0)\ee^{-s\Omega}$ as in~\eqref{eq:RandOmega}, we have
\begin{equation}\label{eq:TrCDssp}
\Tr C(s,s')D(s,s') = -(s-s')^2 [\expc((s-s')K)\eta
\ee^{-(s-s')\Omega}
\ee^{(s-s')K}\eta\ee^{-(s-s')\Omega} \expc((s-s')K)]_{00},
\end{equation}
so $\Tr D(s,s')$ is translationally invariant. If $\kappa>0$, then 
\begin{equation}\label{eq:Jdef}
  \Tr C(s,s')D(s,s') = -(s-s')^2\Jc_\Omega(\kappa^2(s-s')^2)
\end{equation}
for all $s,s'\in\RR$, where $\Jc_\Omega$ is entire analytic, real-valued on $\RR^+$, and has Taylor expansion
\begin{equation}
      \Jc_\Omega(z) = 1+ \frac{z}{12} + \frac{\kappa^2-\tau^2}{360\kappa^2}z^2 +
     \mathcal{O}(z^3). \label{eq:J_exp}
\end{equation}
If $\kappa=0$ then 
\begin{equation}\label{eq:TrCD_kappa0}
     \Tr C(s,s')D(s,s') = -(s-s')^2.
\end{equation}
\end{lemma}
Note that dependence on $\Omega$ only enters at $O(z^3)$ in~\eqref{eq:J_exp}.
 
As a consequence of~\eqref{eq:Jdef},~\eqref{eq:Gdef} and~\eqref{eq:Hdef} is that, for $\kappa>0$,  
\begin{equation}\label{eq:TrCTrDminusTrCD}
    \Tr C(s,s')\Tr D(s,s')- \Tr C(s,s')D(s,s')  = (s-s')^2 \Xi_\Omega(\kappa^2 (s-s')^2), 
\end{equation} 
where the entire analytic function $\Xi_\Omega = \Jc_\Omega- \Gc_\Omega \Hc_\Omega$ is real-valued on $\RR^+$. After simplification, its Taylor expansion about $z=0$ is found to be
\begin{align}
\Xi_\Omega(z) &=   3 + (\kappa^2-4\|\boldsymbol{\varpi}-\boldsymbol{\varpi}^\FW\|^2)\frac{z}{4\kappa^2}
  \nonumber\\ &\qquad 
    + ((3\kappa^2+7\tau^2 -30 \|\boldsymbol{\varpi}-\boldsymbol{\varpi}^\FW\|^2)\kappa^2
    +30\| \boldsymbol{\varpi}-\boldsymbol{\varpi}^\FW\|^4+30
    \|\boldsymbol{\varpi}\times\boldsymbol{\varpi}^\FW\|^2)
    \frac{z^2}{360\kappa^4}
   \nonumber\\ &\qquad   +\mathcal{O}(z^3).
\end{align} 
In the case $\boldsymbol{\varpi}=\boldsymbol{\varpi}^\FW$ we have 
\begin{equation}
    \Xi^\FW(z) = 3 + \frac{z}{4} + \frac{3\kappa^2+7\tau^2}{360\kappa^2}z^2+\mathcal{O}(z^3),
\end{equation}
and in fact the $\mathcal{O}(z^2)$ term has a critical point at $\boldsymbol{\varpi}=\boldsymbol{\varpi}^\FW$.
If $\kappa=0$, we combine~\eqref{eq:TrC_kappa0},~\eqref{eq:TrD_kappa0} and~\eqref{eq:TrCD_kappa0} and use a trigonometric identity to find
\begin{equation}\label{eq:TrCTrDminusTrCD_kappa0}
\Tr C(s,s')\Tr D(s,s')- \Tr C(s,s')D(s,s')  = (s-s')^2\left(3-4\sin^2(\tfrac{1}{2}\|\boldsymbol{\varpi}\|(s-s'))\right).
\end{equation}

\section{QEIs along stationary worldlines}\label{sec:QEIs}

\subsection{Scalar and electromagnetic QEIs for arbitrary frames}

We now compute the point-split energy densities 
$T(s,s')$ and $T_{\textnormal{e.m.}}(s,s')$ using
the calculations in the previous section and the 
identity (see~(4.13) and Appendix~A in~\cite{FewsterThompson2023}) 
\begin{equation}\label{eq:identity}
    \lim_{\epsilon\to 0+} \frac{(s-s')^{2j}}{\sigma_\epsilon(\gamma(s),\gamma(s'))^k}
    = \frac{(-1)^k}{\Upsilon(\kappa^2(s-s')^2)^k} \frac{1}{(s-s'-\ii 0^+)^{2(k-j)}}
\end{equation}
for $k\in\NN$, and $j\in\NN_0$, where $(z-\ii 0^+)^{-r}$
is shorthand for the distributional limit of $(z-\ii \epsilon)^{-r}$ as $\epsilon\to 0+$, and reduces to $z^{-r}$ if $r$ is a nonpositive integer.

In the scalar case, we apply~\eqref{eq:identity} to~\eqref{eq:T_formula_scalar}, combined with~\eqref{eq:Gdef} and~\eqref{eq:Hdef}, to find that $T$ is translationally invariant, 
$T(s,s')=\Tc(s-s')$, with
\begin{equation}\label{eq:Tc_scalar}
    \Tc(s) =  -\frac{\Kc_\Omega(\kappa^2 s^2)}{4\pi^2(s -\ii 0^+)^4},
\end{equation}
where 
\begin{equation}
    \Kc_\Omega(z) = \frac{\Gc_\Omega(z)}{\Upsilon(z)^2} - 4 \frac{\Hc_\Omega(z)}{\Upsilon(z)^3}
\end{equation}
is a meromorphic function, analytic in a neighbourhood of $\RR^+$ and real-valued on $\RR^+$. In the case $\Omega=0$ this agrees with equation~(4.13) of~\cite{FewsterThompson2023}. The Taylor series of $\Kc_\Omega$ at the origin is
\begin{equation}\label{eq:K_scalar}
    \Kc_\Omega(z) = -6 + \frac{A_\Omega}{\kappa^2}z - \frac{B_\Omega}{4\kappa^4}z^2+ \mathcal{O}(z^3),
\end{equation} 
where the constants $A_\Omega$ and $B_\Omega$ may be obtained from the Taylor series for $\Gc_\Omega$, $\Hc_\Omega,$ and $\Upsilon$ as
\begin{align}\label{eq:A_scalar_Omega}
    A_\Omega&= \kappa^2 +\|\boldsymbol{\varpi}-\boldsymbol{\varpi}^\FW\|^2\\ \label{eq:B_scalar_Omega}
    B_\Omega&=  \frac{33\kappa^4+22(\kappa\tau)^2}{90} +\kappa^2\left(\varpi_2^2
    + (\varpi_3-2\tau/3)^2\right) \nonumber\\ &\qquad +\frac{1}{3} \left(
    \|\boldsymbol{\varpi}-\boldsymbol{\varpi}^\FW\|^4+ 2 \kappa^2\|\boldsymbol{\varpi}-\boldsymbol{\varpi}^\FW\|^2 +\|\boldsymbol{\varpi}\times\boldsymbol{\varpi}^\FW\|^2
    \right)
\end{align}
after some manipulation. Again, the results with $\boldsymbol{\varpi}=0$ coincide with those obtained in~\cite{FewsterThompson2023} (see Eqs.~(1.9) and~(1.10))). If $\kappa=0$, one finds 
from~\eqref{eq:TrC_kappa0} and~\eqref{eq:TrD_kappa0} (also noting that $\Upsilon\equiv 1$) that
\begin{equation}\label{eq:Tc_kappa0}
    \Tc(s) =  \frac{4+2\cos(\|\boldsymbol{\varpi}\|s)}{4\pi^2(s-\ii 0^+)^4}.
\end{equation}

In the electromagnetic case, we apply~\eqref{eq:identity} to~\eqref{T_formula}, combined with~\eqref{eq:TrCTrDminusTrCD}, to find that $T_{\textnormal{e.m.}}$ is translationally invariant, $T_{\textnormal{e.m.}}(s,s')=\Tc_\textnormal{e.m.}(s-s')$, with
\begin{equation}\label{eq:Tc_em}
    \Tc_{\textnormal{e.m.}}(s) = 
    -\frac{\Kc_{\textnormal{e.m.},\Omega}(\kappa^2s^2)}{4\pi^2(s-\ii0^+)^4},
\end{equation}
where
\begin{equation}
\Kc_{\textnormal{e.m.},\Omega}(z) =-4\frac{\Xi_\Omega(z)}{\Upsilon(z)^3}
\end{equation}
is meromorphic on $\CC$, analytic in a neighbourhood of $\mathbb{R}^+$ and real-valued on $\RR^+$. For $\kappa>0$, $\Kc_{\textnormal{e.m.},\Omega}$ has Taylor series 
\begin{equation}
    \Kc_{\textnormal{e.m.}}(z)   = -12 + \frac{2A_{\textnormal{e.m.},\Omega}}{\kappa^2}z - \frac{B_{\textnormal{e.m.},\Omega}}{2\kappa^4} z^2  +\mathcal{O}(z^3)
\end{equation}
about $z=0$, with constants that are obtained from the Taylor series of $\Xi_\Omega$ and $\Upsilon$ as
\begin{align}\label{eq:A_em_Omega}
    A_{\textnormal{e.m.},\Omega} &=\kappa^2+2\|\boldsymbol{\varpi}-\boldsymbol{\varpi}^\FW\|^2 \\ \label{eq:B_em_Omega}
    B_{\textnormal{e.m.},\Omega} &= \frac{ 
33\kappa^4+32(\kappa\tau)^2}{90}  +\frac{2}{3}\left(
 \|\boldsymbol{\varpi}-\boldsymbol{\varpi}^\FW\|^4
+2 \kappa^2 \|\boldsymbol{\varpi}-\boldsymbol{\varpi}^\FW\|^2
+  \|\boldsymbol{\varpi}\times\boldsymbol{\varpi}^\FW\|^2
    \right).
\end{align}
The rotational terms in $B_{\textnormal{e.m.},\Omega}$ are minimised where $\boldsymbol{\varpi}-\boldsymbol{\varpi}^\FW$, whereas $B_{\Omega}$ contains a term that is not, ultimately deriving from the $\mathcal{O}(z^2)$ terms in $\Gc_\Omega$ and $\Hc_\Omega$. In the electromagnetic case, they appear multiplied together within $\Xi_\Omega$, and the problematic term cancels out in the $\mathcal{O}(z^2)$ term of $\Gc_\Omega\Hc_\Omega$. We will discuss the optimisation problem further in Section~\ref{sec:optimisation}. If $\kappa=0$, one finds from~\eqref{eq:TrCTrDminusTrCD_kappa0} that
\begin{equation}\label{eq:Tc_em_kappa0}
    \Tc_{\textnormal{e.m.}}(s) = \frac{3-4\sin^2(\tfrac{1}{2}\|\boldsymbol{\varpi}\|(s-s'))}{\pi^2(s-\ii 0^+)^4}.
\end{equation}
 
To simplify equations~\eqref{eq:Tc_scalar},~\eqref{eq:Tc_kappa0},~\eqref{eq:Tc_em}, and~\eqref{eq:Tc_em_kappa0}, 
we use the fact that
\begin{equation}\label{eq:Fovers4}
     \frac{F(\kappa^2 s^2)}{(s-\ii 0^+)^{4}}
=  \frac{F(0)}{(s-\ii 0^+)^{4}} + 
 \frac{\kappa^2 F'(0)}{(s-\ii 0^+)^{2}} + 
\kappa^{4} F_{\textnormal{cc}}(\kappa^2 s^2)
\end{equation}
for any $F\in C^\infty(\RR)$, where the double cardinal function $F_{\textnormal{cc}}$ is the cardinal function of $F_{\textnormal{c}}$,  
\begin{equation}
    F_{\textnormal{cc}}(z) = \begin{cases}
        z^{-2}(F(z)-F(0)-zF'(0)) & z \neq 0\\ \frac{1}{2}F''(0) & z=0.
    \end{cases}
\end{equation}
The identity~\eqref{eq:Fovers4} is proved by writing $F(z)=F(0)+F'(0)z+F_\textnormal{cc}(z)z^2$ and expanding the left-hand side, using $s^{2j}/(s-\ii 0^+)^4= 1/(s-\ii 0^+)^{4-2j}$, which reduces to 
$s^{2j-4}$ if $j\ge 2$. Note that if $F(z)=\mathcal{O}(z)$ as $z\to+\infty$ on the real line, then $F_{\textnormal{cc}}(z)=\mathcal{O}(z^{-1})$ in this limit, and therefore the last term in \eqref{eq:Fovers4} has order $\mathcal{O}(s^{-2})$ as $|s|\to\infty$ in $\RR$. This is true for both $\Kc_{\Omega}$ and $\Kc_{\textnormal{e.m.},\Omega}$ due to the following result, established in Appendix~\ref{appx:asymptotic}.
\begin{lemma}\label{lem:K_asympt} 
As $|s|\to\infty$ with $s\in\RR$, the 
functions $\Kc_\Omega$ and $\Kc_{\textnormal{e.m.},\Omega}$ satisfy 
    \begin{equation}\label{eq:K_scalar_asympt}
    \Kc_{\Omega}(\kappa^2s^2)=  
    \begin{cases}  
    \mathcal{O}(1) & \text{inertial, elliptic worldlines}\\ 
    \mathcal{O}(s^4\ee^{-\chi |s|}) & \text{hyperbolic worldlines}\\ 
         \mathcal{O}(s^{-2}) & \text{parabolic worldlines}
    \end{cases} 
\end{equation} 
and
\begin{equation}\label{eq:K_em_asympt}
    \Kc_{\textnormal{e.m.},\Omega}(\kappa^2s^2)= 
    \begin{cases}  
    \mathcal{O}(1) & \text{inertial, elliptic worldlines}\\ 
    \mathcal{O}(s^5\ee^{-\chi |s|}) & \text{hyperbolic worldlines}\\ 
         \mathcal{O}(s^{-1}) & \text{parabolic worldlines.}
    \end{cases}
\end{equation}
\end{lemma}
In the scalar case with $\kappa>0$, we have thus established the decomposition~\eqref{eq:Tc_general_form_scalar} with $A$ replaced by $A_\Omega$ and where
\begin{equation}
\Tc_\reg(s)= -\frac{\kappa^4}{4\pi^2} (\Kc_\Omega)_\textnormal{cc}(\kappa^2 s^2)
\end{equation}
is smooth, real and even on $\RR$, with $\mathcal{O}(s^{-2})$ fall-off as $|s|\to\infty$ and $\Tc_\reg(0)=B_\Omega/(16\pi^2)$. For $\kappa=0$,~\eqref{eq:Tc_kappa0} gives~\eqref{eq:Tc_general_form_scalar} with
$A_\Omega = \|\boldsymbol{\varpi}\|^2$ and $\Tc_\reg(s) =(2\pi^2)^{-1} \cos_\textnormal{cc}(\|\boldsymbol{\varpi}\|s)$, also being smooth real and even on $\RR$, with $\mathcal{O}(s^{-2})$ decay, and $B=16\pi^2\Tc_\reg(0)=\|\boldsymbol{\varpi}\|^4/3$. Note that
these values for $A_\Omega$ and $B_\Omega$ coincide with the $\kappa=0$ (and hence $\tau=\upsilon=0$, $\boldsymbol{\varpi}^\FW=\boldsymbol{0}$) instances of~\eqref{eq:A_scalar_Omega},~\eqref{eq:B_scalar_Omega}.

Similarly, in the electromagnetic case, $\Tc_\textnormal{e.m.}$ decomposes as in~\eqref{eq:Tc_general_form_em}, with
\begin{equation}
\Tc_{\textnormal{e.m.},\reg}(s)= -\frac{\kappa^4}{4\pi^2} (\Kc_{\textnormal{e.m.},\Omega})_\textnormal{cc}(\kappa^2 s^2)
\end{equation}
being smooth, real and even on $\RR$, with $\mathcal{O}(s^{-2})$ fall-off as $|s|\to\infty$ and $\Tc_{\textnormal{e.m.},\reg}(0)=B_{\textnormal{e.m.},\Omega}/(8\pi^2)$.
For $\kappa=0$,~\eqref{eq:Tc_em_kappa0} also gives a decomposition as in~\eqref{eq:Tc_general_form_em} satisfying the requirements, and with $A_{\textnormal{e.m.},\Omega}=2\|\boldsymbol{\varpi}\|^2$, 
$B_{\textnormal{e.m.},\Omega}=\frac{2}{3}\|\boldsymbol{\varpi}\|^4$, agreeing with the 
$\kappa=0$ instances of~\eqref{eq:A_em_Omega},~\eqref{eq:B_em_Omega}.

Consequently, in all cases, we have derived the massless scalar QEI~\eqref{eq:stationary_scalar_QEI} with constants $A$, $B$ replaced by $A_\Omega$, $B_\Omega$ from~\eqref{eq:A_scalar_Omega},~\eqref{eq:B_scalar_Omega}, and the electromagnetic QEI~\eqref{eq:stationary_em_QEI} with constants $A_\textnormal{e.m.}$, $B_\textnormal{e.m.}$ replaced by $A_{\textnormal{e.m.},\Omega}$, $B_{\textnormal{e.m.},\Omega}$ from~\eqref{eq:A_em_Omega},~\eqref{eq:B_em_Omega}.

\subsection{Optimising the QEIs}\label{sec:optimisation}

The QEIs obtained in the previous subsection depend on the angular velocity of the
frame relative to the Frenet--Serret tetrad. As the energy density must obey all of them, we may optimise over $\boldsymbol{\varpi}$ in order to find the tightest QEI bound in this class. 

The electromagnetic case is straightfoward, because 
it is clear from~\eqref{eq:A_em_Omega} and~\eqref{eq:B_em_Omega}
that $A_{\textnormal{e.m.},\Omega}$ and 
$B_{\textnormal{e.m.},\Omega}$ both attain their global minima at $\boldsymbol{\varpi}=\boldsymbol{\varpi}^\FW$, i.e., using a Fermi--Walker frame, and yielding the values in~\eqref{eq:ABFW}.

In the scalar case,~\eqref{eq:A_scalar_Omega} shows that $A_\Omega$ attains a global minimum $A^\FW=\kappa^2$ at $\boldsymbol{\varpi}=\boldsymbol{\varpi}^\FW$, which compares favourably with the Frenet--Serret value $A^\FS=\kappa^2+\tau^2+\upsilon^2$ found in~\cite{FewsterThompson2023}.
Surprisingly, however, $B_\Omega$ is not globally minimised in the Fermi--Walker frame unless $\tau=0$. This is a consequence of the failure of the $\mathcal{O}(z^2)$ terms in $\Gc_\Omega$ and $\Hc_\Omega$ to have critical points where $\boldsymbol{\varpi}=\boldsymbol{\varpi}^\FW$. Indeed $\partial B_\Omega/\partial\varpi_3=2\kappa^2\tau/3$ at $\boldsymbol{\varpi}=\boldsymbol{\varpi}^\FW$, where
\begin{equation}
    B^\FW = \frac{\kappa^2(33\kappa^2+ 32\tau^2)}{90}
\end{equation}
(distinctly smaller than the Frenet--Serret value $B^\FS$ given in~\eqref{eq:AB_FS}).
Nonetheless,~\eqref{eq:B_scalar_Omega} exhibits $B_\Omega$ as a constant plus a sum of squares depending on $\boldsymbol{\varpi}$ (though not simultaneously vanishing unless $\tau=0$) thus providing a lower bound on its global minimum value $B_\textnormal{min}$. On the other hand, $B_\textnormal{min}$ is bounded above by the value of $B_\Omega$ at $\boldsymbol{\varpi}=\boldsymbol{\varpi}^\FW$. If $\tau>0$, we deduce that $B_\textnormal{min}$ obeys
\begin{equation}\label{eq:B_bounds}
    \frac{\kappa^2(33\kappa^2+22\tau^2)}{90} < B_\textnormal{min} <
    \frac{\kappa^2(33\kappa^2 + 32\tau^2)}{90},
\end{equation}
while the global minimum is $B_\textnormal{min}=11\kappa^4/30$ if $\tau=0$, i.e., for the cases of inertial motion or uniform linear acceleration. In the light of~\eqref{eq:B_bounds}, a useful figure of merit is 
\begin{equation}
    \beta = \frac{90 B_\textnormal{min}-33\kappa^4}{\kappa^2\tau^2},
\end{equation}
i.e., the effective coefficient of $(\kappa\tau)^2$ in $B_\textnormal{min}$,
which is bounded by $22<\beta<32$ for $\tau>0$ by~\eqref{eq:B_bounds}.

More detailed analysis, explained in Appendix~\ref{appx:optimisation},
shows that $\beta$ can be expressed solely in terms of the
dimensionless parameters $t=\tau/\kappa$ and $u=\upsilon/\kappa$, and obeys the sharpened bound $26<\beta<32$, with the lower limit approached as $t=\tau/\kappa$ and $u=\upsilon/\kappa$ tend to zero, and the upper bound approached, for example, as $|\upsilon|\to\infty$ with $\kappa,\tau$ fixed. In the former regime,
one finds
\begin{equation}\label{eq:beta_small}
    \beta = 26+ \frac{6}{125}(t^2+25u^2) + \mathcal{O}((t^2+u^2)^{2}) 
\end{equation}
and hence
\begin{align}\label{eq:Bmin_small}
    B_\textnormal{min}&=\frac{1}{90}\left(
33 \kappa^{4}+26 \kappa^{2} \tau^{2}+\frac{6 \tau^{4}}{125}+\frac{6 \tau^{2} \upsilon^{2}}{5}
\right)
+\mathcal{O}\left(\frac{(\tau^2+\upsilon^2)^3}{\kappa^2}\right)
\end{align}
as $\tau/\kappa,\upsilon/\kappa\to 0$. The minimum is achieved for $\boldsymbol{\varpi}$ given by 
\begin{equation}\label{eq:optvarpi_small}
   \varpi_1=\upsilon\left(1+\mathcal{O}\left(\frac{\tau^2+\upsilon^2}{\kappa^2} \right)\right),\qquad 
   \varpi_2=0,\qquad 
   \varpi_3=\frac{4\tau}{5}\left(1+\mathcal{O}\left(\frac{\tau^2+\upsilon^2}{\kappa^2} \right)\right),
\end{equation}
in which the value of $\varpi_2$ is exact.
For example, parabolic stationary worldlines correspond to $t=1$, $u=0$. Numerically, we find
$\varpi_1=0$ and $\varpi_3=0.8030555623\tau$, and
$\beta= 26.04653366$, while whereas the approximation just given would give $\varpi_1=0$, $\varpi_3=0.8\tau$ and $\beta=26.048$.

In the regime where $\tau^2+\upsilon^2\to\infty$ with $\tau/|\upsilon|$ bounded away from zero, it is shown in Appendix~\ref{appx:optimisation} that the minimum is achieved for
\begin{equation}\label{eq:varpi_ultraloxo} 
    \varpi_1 \sim \upsilon\left(1 -\frac{(\kappa\tau/\sqrt{2})^{2/3}}{(\tau^2+\upsilon^2)^{2/3}} \right) ,\qquad  \varpi_2=0,\qquad
    \varpi_3 \sim \tau\left(1 -\frac{(\kappa\tau/\sqrt{2})^{2/3}}{(\tau^2+\upsilon^2)^{2/3}}\right),
\end{equation}
giving
\begin{equation}\label{eq:beta_ultraloxo}
    \beta\sim 32 -45 \left(\frac{\kappa\tau/\sqrt{2}}{\tau^2+\upsilon^2}\right)^{2/3}
\end{equation}
as $|u|\to\infty$ with $t>0$ constant. Thus
\begin{equation}\label{eq:Bmin_ultraloxo}
    B_\textnormal{min}\sim \frac{\kappa^2}{90}
    \left(33\kappa^2 + 32\tau^2  -45 \tau^2\left(\frac{\kappa\tau/\sqrt{2}}{\tau^2+\upsilon^2}\right)^{2/3}
    \right).
\end{equation}

\begin{figure}
    \centering
    \includegraphics[width=0.45\linewidth]{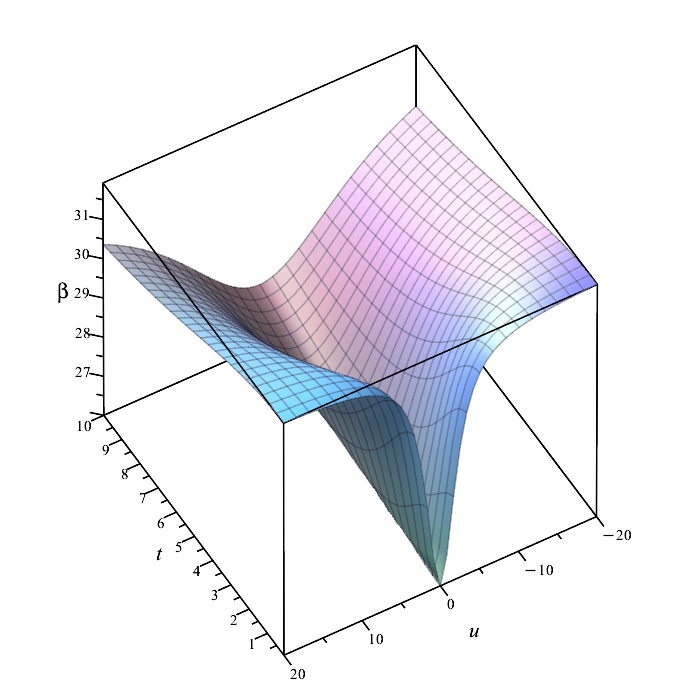}\hfil
    \includegraphics[width=0.45\linewidth]{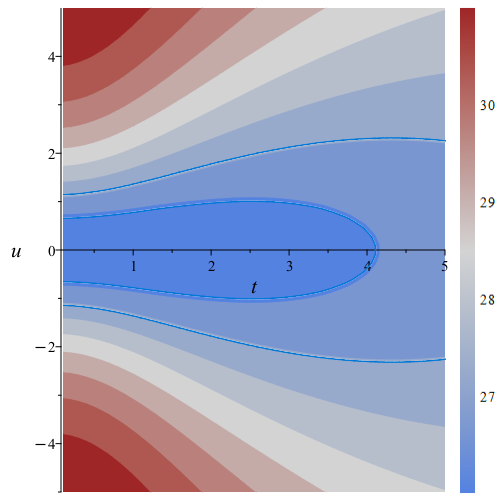}
    \caption{Plots of $\beta = (90 B_\textnormal{min}-33\kappa^4)/(\kappa\tau)^2$ against $t=\tau/\kappa$ and $u=\upsilon/\kappa$. Left-hand image: a three-dimensional plot illustrating the fact that $\beta\to 26$ as $t,u\to 0$ and $\beta\to 32$ as $t$ and $u$ become large. Right-hand image: a heatmap plot of $\beta$ showing the existence of a region near $t=u=0$ in which $\beta\approx 26$.}
    \label{fig:min}
\end{figure}

Plots of $\beta$, obtained by numerical minimisation, against $t=\tau/\kappa$ and $u=\upsilon/\kappa$ are given in Fig.~\ref{fig:min}, and clearly indicates the behaviours just described for small $t, u$ and $t,u\to \infty$ with fixed $t/|u|$ bounded away from zero. From the heatmap (right-hand image) one also can see that the approximation $\beta\approx 26$ is valid over a region near $t=u=0$ that encompasses all of the catenary worldlines ($0<t<1$, $u=0$), the parabolic worldlines ($t=1$, $u=0$) and portions of the elliptic ($t>1$, $u=0$) and loxodromic ($t,u>0$) parameter space. Some more details concerning these plots are given in Appendix~\ref{appx:optimisation}.
   
The value of $A$ is not minimised for the frame that optimises $B$; instead we have $A\approx \kappa^2+\tau^2/25$ for small $\tau/\kappa,\upsilon/\kappa$, or
\begin{equation}
    A\approx \kappa^2 + \frac{(\kappa\tau/\sqrt{2})^{4/3}}{(\tau^2+\upsilon^2)^{1/3}}
\end{equation}
in the limit $\tau^2+\upsilon^2\to\infty$ with $\tau/|\upsilon|$ bounded away from zero. However, it is close to its minimum value in each case.

\section{Example: loxodromic worldlines}\label{sec:loxodromic}

In this section we compute the energy density of a specific state along the loxodromic worldline given by
\begin{equation}\label{eq:loxo}
    \gamma(s) = \left(\frac{C}{\chi}\sinh(\chi s),
    \frac{C}{\chi}\cosh(\chi s),\frac{V}{\omega}\cos(\omega s),\frac{V}{\omega}\sin(\omega s)\right),
\end{equation}
(see Section~5.6 of~\cite{FewsterThompson2023} for the corresponding Frenet--Serret tetrad) where $\chi$ and $\omega$ are related to the curvatures as in~\eqref{eq:chi} and~\eqref{eq:omega}, while $C>1$ and $V\in\RR$ are related by $C^2-V^2=1$ and $C$ is related to the curvatures by 
\begin{equation}
C^2 = \frac{\kappa^2+\tau^2+\upsilon^2+R}{2R},
\end{equation}
with $R$ as in~\eqref{eq:R}.
One has
\begin{equation}
\kappa = \sqrt{C^{2} \chi^{2}+V^{2} \omega^{2}},\qquad
\tau = 
\frac{\left(\chi^{2}+\omega^{2}\right) C |V|}{\sqrt{C^{2} \chi^{2}+V^{2} \omega^{2}}},\qquad
\upsilon = \frac{\chi  \omega}{\sqrt{C^{2} \chi^{2}+V^{2} \omega^{2}}}.
\end{equation}
Recall that all stationary worldlines with the same curvatures can be transformed into one another by proper orthochronous Poincar\'e transformations, so there is no loss of generality in considering loxodromic 
worldlines of the form~\eqref{eq:loxo}. An important property of this curve is that its causal hull is the open right wedge
\begin{equation}\label{eq:cvxhull}
\bigcup_{s_1,s_2\in\RR} J^+(\gamma(s_1))\cap J^-(\gamma(s_2))=    \Wc: = \{(t,x,y,z)\in\RR^4: x>|t|\}.
\end{equation}
For it is clear that the left-hand side is a subset of the right, because $\gamma(s)\in\Wc$ and $\Wc$ is causally convex. Conversely, if $p\in\Wc$ then 
\begin{align}
    \eta(p-\gamma(s),p-\gamma(s)) &= 
    t^2-x^2+\left(\frac{C}{\chi}\right)^2
    +\frac{C}{\chi}\left((x-t)\ee^{\chi s}+(x+t)\ee^{-\chi s}\right) \nonumber\\ &\qquad
    -(y-V\omega^{-1}\cos\omega s)^2 -(z-V\omega^{-1}\sin\omega s)^2,
\end{align}
where $p=(t,x,y,z)$, and this expression tends to $+\infty$ as $|s|\to\infty$ because $x>|t|$. Thus $p$ is eventually in $J^+(\gamma(s))$ as $s\to-\infty$, and $J^-(\gamma(s))$ as $s\to+\infty$, establishing~\eqref{eq:cvxhull}

Any Hadamard state of the scalar or electromagnetic fields on $\Wc$ must obey the Minkowski QEIs along any curve contained within $\Wc$. This is because the QEIs are local and covariant as discussed in~\cite{QEI2006FewsterPfenning}; see also the discussion in Section~6 of~\cite{FewsterThompson2023}. In particular, this applies to the Fulling--Rindler vacuum, i.e., the ground state for the boost Killing vector associated with $\Wc$, for which the renormalised stress-energy tensor is known in both the scalar and electromagnetic cases. The result for the massless minimally coupled scalar field is\footnote{See, e.g., Appendix D of~\cite{FewsterThompson2023} which uses the Wightman functions derived in~\cite{MorettiVanzo:1996}; the calculation is standard, but note that many references follow~\cite{CandelasDeutsch:1977} and consider the conformally coupled stress-energy tensor rather than the minimally coupled one used here.} 
\begin{equation}
    \langle T_{\mu\nu}\rangle = -\frac{1}{1440\pi^2\xi^4}\left(14v_\mu v_\nu+ 30 w_\mu w_\nu+ 19\eta_{\mu\nu}\right),
\end{equation} 
while the electromagnetic result was computed by Candelas and Deutsch~\cite{CandelasDeutsch:1977} as
\begin{equation}
    \langle T_{\mu\nu}\rangle_\textnormal{e.m.} =  -\frac{11}{720\pi^2\xi^4}\left(4v_\mu v_\nu-\eta_{\mu\nu}\right).
\end{equation} 
In each case, $\xi=\sqrt{x^2-t^2}$, while 
\begin{equation}
    v^\bullet =\frac{1}{\sqrt{x^2-t^2}} (x,t,0,0),\qquad
    w^\bullet =\frac{1}{\sqrt{x^2-t^2}} (t,x,0,0),
\end{equation}
are unit vectors parallel to the $4$-velocity and $4$-acceleration of the curve through $(t,x,y,z)$ on which $\xi$, $y$, and $z$ are all constant. 
For the worldline~\eqref{eq:loxo} with $u=\dot{\gamma}$, one finds that $v_\lambda u^\lambda = C$, $w_\lambda u^\lambda=0$, and $\xi=C/\chi$, so the energy density measured along the worldline~\eqref{eq:loxo} takes the constant value
\begin{equation}
\langle \rho\rangle = \langle T_{\mu\nu}u^\mu u^\nu\rangle = -\frac{(14 C^2+19)\chi^4}{1440\pi^2 C^4}
\end{equation}
in the scalar case and
\begin{equation}\label{eq:rhoem}
\langle \rho\rangle_\textnormal{e.m.} = \langle T_{\mu\nu}u^\mu u^\nu\rangle_\textnormal{e.m.} =  -\frac{(44 C^2-11)\chi^4}{720\pi^2 C^4}
\end{equation}
for electromagnetism. These exact values can be compared with the lower bounds $-B_\textnormal{min}/(16\pi^2)$ and $-B^\FW_\textnormal{e.m.}/(8\pi^2)$ given by the QEIs, as a consistency check.

As $C^2\chi^2\le \kappa^2$ and hence $\chi^4/C^4\le \kappa^4/C^8$ one has
\begin{equation}
\langle \rho\rangle \ge    
-\frac{33\kappa^4}{1440\pi^2} \frac{14C^2+19}{33C^8}\ge
    -\frac{33\kappa^4}{1440\pi^2}
\end{equation}
and
\begin{equation}
    \langle \rho\rangle_\textnormal{e.m.} \ge -\frac{33\kappa^4}{720\pi^2} \frac{44C^2-11}{33C^8}\ge
    -\frac{33\kappa^4}{720\pi^2},
\end{equation}
because $C^2\ge 1$. Comparing with~\eqref{eq:constantbound} and~\eqref{eq:constantbound_em}, we see that both QEIs are satisfied for all values of the curvature parameters, and with room to spare except in the limit $\tau\to 0$, in which they become sharp. This is the expected behaviour, because it is the limit in which the curve becomes a uniformly accelerated trajectory.

The fact that any stationary worldline is an orbits of a timelike Killing vector field was used by Letaw and Pfautsch~\cite{LetawPfautsch:1982} to construct an associated coordinate system covering a patch of Minkowski spacetime. The loxodromic curve~\eqref{eq:loxo} belongs to `class F' in~\cite{LetawPfautsch:1982}, and is associated with rotating pseudocylindrical coordinates. These coordinates are used in~\cite{LetawPfautsch:1981} to 
construct a stationary state for the minimally coupled Klein--Gordon equation with mass $M\ge 0$ by canonical quantisation, resulting in the Fulling--Rindler vacuum associated with the wedge $\Wc$. From this perspective it is clear from the start that the energy density of the Fulling--Rindler vacuum will be constant along the loxodromic worldlines. However, it is an open question whether the Fulling--Rindler vacuum is the Hadamard state with the lowest possible constant energy density along the loxodromic curves. If it is, then the results presented above would indicate that the QEI is not sharp, even in the limit of long sampling times.

\section{Conclusion}\label{sec:conclusion}

We have computed closed form expressions for the quantum energy inequalities for the massless scalar and electromagnetic fields along stationary worldlines in Minkowski spacetime, using tetrads rotating with arbitrary fixed angular velocity relative to the Frenet--Serret frame. For vanishing angular velocity, this reproduces the results of~\cite{FewsterThompson2023} for the scalar field. Our method differs from that of~\cite{FewsterThompson2023} in that we focus on the Frenet--Serret description of the stationary worldlines and obtain many of our results using functional calculus applied to the matrix of curvature invariants $K$. By contrast,~\cite{FewsterThompson2023} only considered the Frenet--Serret frame, and characterised the worldlines using the evolution of the velocity vector under a $1$-parameter Lorentz subgroup, using ad hoc arguments to deduce that various quantities in the point-split energy density, such as the traces $\Tr C(s,s')$ and $\Tr D(s,s')$ are translationally invariant. Our current method is more general, systematic and direct, and may also be useful in other calculations for QFT in relation to stationary worldlines which have traditionally proceeded on a case-by-case basis (see, e.g.,~\cite{Letaw:1981,LetawPfautsch:1981,GoodJuarezAubryetal:2020} and also~\cite{FewsterThompson2023}) but which is largely avoided here. 
The partial exception is the asymptotic behaviour discussed in Appendix~\ref{appx:asymptotic}, where results divide into three cases, covering the inertial and elliptic worldlines, the hyperbolic worldlines, and the parabolic worldline.

We were then able to optimise over the angular velocity to find the strongest possible QEI among these possibilities. For electromagnetism, the optimal bound arises precisely for tetrads invariant under Fermi--Walker transport along the worldline. The QEI obtained in this case is exactly twice the QEI for massless scalars, also using a Fermi--Walker tetrad. These are natural results, given the two polarisation degrees of freedom in electromagnetism, and it is also natural to employ a Fermi--Walker tetrad, which represents the rods and clocks of a nonrotating laboratory following the worldline. Indeed, this choice was proposed long ago in~\cite{QEI2006FewsterPfenning} as a means to ensure that the QEIs are locally covariant. In this light it is somewhat surprising to discover, as we have shown, that the scalar QEI is not optimised for Fermi--Walker tetrads, although it is a good approximation at high values of torsion $\tau$ and hypertorsion $\upsilon$ relative to the curvature $\kappa$. When $\tau/\kappa$ and $\upsilon/\kappa$ are small, however, the constant term in the QEI is minimised when the tetrad rotates with angular velocity approximately $(\upsilon,0,4\tau/5)$ relative to the Frenet--Serret tetrad, while the the Fermi--Walker tetrad has angular velocity  $(\upsilon,0,\tau)$. This regime includes the catenary and parabolic curves, and regions within the parameter space of elliptic and loxodromic curves (see Fig.~\ref{fig:min}).

One might wonder why the QEI depends on the frame at all. After all, the energy density along the worldline depends on its four-velocity but not the spatial parts of the tetrad. The reason is that the QEIs of~\cite{GenQEI2000CFewster,FewsterPfenning:2003} proceed by decomposing the energy density as a sum of squares using the tetrad, and then summing a QEI for each term separately. In other words, one assumes the worst case for each component, whereas in the actual energy density negative fluctuations in one component can be compensated by positive fluctuations in another. The component QEI bounds are increasingly negative when the tetrad rotates relative to the Fermi--Walker case, as seen from the expressions~\eqref{eq:A_scalar_Omega},~\eqref{eq:B_scalar_Omega},~\eqref{eq:A_em_Omega}, and~\eqref{eq:B_em_Omega} for $A_\Omega$, $B_\Omega$, $A_{\textnormal{e.m.},\Omega}$, and $B_{\textnormal{e.m.},\Omega}$. One may think of this as a type of Unruh effect: 
a detector coupled to a quantum field with a dipole moment can be excited not only by the detector's acceleration, but also by the angular velocity of the dipole around the worldline.  

The difference between the scalar and electromagnetic fields can be seen in two ways. First, the number of polarisations enters as a constant multiplying some common rotational terms. Second, the scalar case exhibits additional terms proportional to $\varpi_2^2$ and $(\varpi_3-2\tau/3)^2$ in  $B_{\Omega}$, and the contribution of $(\kappa\tau)^2$ enters with a different coefficient. As the Fermi--Walker tetrad has $\varpi_2=0$, it is the term in $(\varpi_3-2\tau/3)^2$ that is responsible for moving the global minimum of $B_{\Omega}$ away from the Fermi--Walker tetrad. This term has its origins in the structure of the stress-energy tensors and point-split energy densities. The problematic term in the scalar field case can be traced to the $\mathcal{O}(z^2)$ terms in $\Gc_\Omega$ and $\Hc_\Omega$; but in the electromagnetic case the problematic parts cancel out, due to Lemma~\ref{lem:trace_squared_C} and the fact that the remaining contributions involve the product $\Gc_\Omega\Hc_\Omega$ rather than the combination $\Gc_\Omega\Upsilon-4\Hc_\Omega$ appearing in the scalar case. However, we do not have a simple high-level explanation of the difference.

In any case, by optimising the QEI bounds, we remove the rotational contributions that are present if the Frenet--Serret tetrad is considered as in~\cite{FewsterThompson2023}, resulting in much more stringent QEIs. While the Frenet--Serret value $B^\FS$ grows with the fourth power of both torsion and hypertorsion, the bounds obtained here are independent of hypertorsion and grow only quadratically in the torsion. As described in the introduction, this more severely restricts the values of constant energy densities that can be sustained along stationary worldlines, and also the timescales on which the energy density can fluctuate below these bounds. It remains open, however, as to whether there really are states that come close to these bounds, except for the inertial and uniformly linearly accelerated worldlines. The attempt by Letaw and Pfautsch~\cite{LetawPfautsch:1981} to canonically associate ground states of the scalar field with coordinate systems based on stationary worldlines resulted in either the Minkowski vacuum, for those worldlines without horizons (inertial, elliptic and parabolic), or a Fulling--Rindler vacuum state on an appropriate wedge otherwise (the three hyperbolic cases). However this does not exclude the possibility that Hadamard states might exist on portions of Minkowski spacetime that saturate the bound in a limit, or simply that there are non-ground states with large constant negative energy densities. As a consistency check, however, we have verified that the QEIs are satisfied along loxodromic curves in the Fulling--Rindler vacuum state determined by the associated wedge region. 
 
Finally, we mention some possible open directions. By similar methods it should be possible to compute analogous QEIs for massless scalar or suitable $p$-form fields~\cite{Pfenning:2009} for stationary worldlines in any even-dimensional Minkowski spacetime and possibly de Sitter spacetimes -- see~\cite{Bunney:2024} for a classification of stationary worldlines in higher dimensions. It would also be possible to study the analogous problem for massless Dirac and Rarita--Schwinger fields using the QEIs obtained in~\cite{FewsterVerch:2001,FewsterMistry:2003,DawsonFewster:2006} and~\cite{YuWu:2004}.
Lastly, an extension to massive fields would be of interest but would certainly require a more significant change to the method.

\paragraph{Acknowledgements} 
CJF thanks Aron Wall for raising the question
of whether the spin of a field might change the QEI along stationary worldlines, and Benito Ju\'arez-Aubry for a stylistic comment. The work of CJF is supported in part by EPSRC Grant EP/Y000099/1 to the University of York.  
For the purpose of open access, the authors have applied a creative commons attribution (CC BY) licence to any author accepted manuscript version arising.
 
\appendix
\section{Derivation of inequality~\eqref{eq:eigenproblem_estimate}}
\label{appx:eigenvalues}
 
Consider a QEI of the general form
\begin{equation}\label{eq:pattern_QEI}
     \int \dd s\, |g(s)|^2\, \langle\rho(\gamma(s))\rangle_\omega \ge -\frac{N}{16\pi^2}\left( \|g''\|^2 + 2A\|g'\|^2 + B\|g\|^2\right)
\end{equation}
for constants $A$, $B$ and $N=1$ for scalars or $N=2$ for electromagnetism, holding for all Hadamard states of the relevant theory and all $g\in C_0^\infty(\RR)$. We will prove a generalised version of~\eqref{eq:eigenproblem_estimate}, drawing on~\cite{FewsterTeo:2000,QEI2006FewsterPfenning}.

Let $I\subset\RR$ be any bounded open interval, with length $T$, and restrict attention to $g\in C_0^\infty(I)$. Then~\eqref{eq:pattern_QEI} implies that
\begin{equation}
    \|g\|^2 \sup_{s\in I} \langle\rho(\gamma(s))\rangle_\omega \ge \int_{I} \dd s\, |g(s)|^2\, \langle\rho(\gamma(s))\rangle_\omega \ge -\frac{N}{16\pi^2}\left( \|g''\|^2 + 2A\|g'\|^2 + B\|g\|^2\right)
\end{equation}
holds for all Hadamard states $\omega$ and all $g\in C_0^\infty(I)$. Rearranging, this implies that
\begin{equation}
     \sup_{s\in I} \langle\rho(\gamma(s))\rangle_\omega \ge -\frac{NB}{16\pi^2} -\frac{N}{16\pi^2}\inf_g \left(\|g''\|^2 + 2A\|g'\|^2\right) 
\end{equation}
with the infimum taken over $g\in C_0^\infty(I)$ with $\|g\|=1$. By standard Rayleigh--Ritz arguments, the infimum equals the minimum eigenvalue $\mu_I(A)$ of the differential operator 
\[
\Lc_I = \frac{\dd^4}{\dd s^4} - 2A\frac{\dd^2}{\dd s^2}
\]
on $I$, subject to boundary conditions $g=g'=0$ at $\partial I$. More precisely, $\mu_I(A)$ is the infimum of the spectrum of the self-adjoint operator formed as the Friedrichs extension of $\Lc_I$ on $C_0^\infty(I)\subset L^2(I,\dd s)$, whose domain corresponds to the boundary conditions just mentioned. By a change of variables argument we see easily that $\mu_I(A) = \mu(A T^2)/T^4$, where $\mu(a)$ is the minimum eigenvalue (in the same sense as before) of $\dd^4/\dd x^4 -2a\dd^2/\dd x^2$ on $(-\tfrac{1}{2},\tfrac{1}{2})$ with boundary conditions $g(\pm \tfrac{1}{2})=g'(\pm \tfrac{1}{2})=0$.  
The $N=1$ case is then precisely the inequality~\eqref{eq:eigenproblem_estimate}.

By solving the ordinary differential equation $g''''-2Ag''-\lambda^2 g=0$ 
on $(-\tfrac{1}{2},\frac{1}{2})$, subject to $g=g'=0$ at $\pm\tfrac{1}{2}$, one finds that eigenfunctions are either even or odd, with the eigenvalues of the even eigenfunctions given by those $\lambda^2>0$ solving
\begin{equation}\label{eq:evaleqn}
    \tanh  \left(\frac{\sqrt{a +\sqrt{a^{2}+\lambda^{2}}}}{2}\right)+\frac{\sqrt{\sqrt{a^{2}+\lambda^{2}}-a}}{\sqrt{a +\sqrt{a^{2}+\lambda^{2}}}}
     \tan  \left(\frac{\sqrt{\sqrt{a^{2}+\lambda^{2}}-a}}{2}\right)=0.
\end{equation}
The eigenvalues of the odd eigenfunctions interleave those of the even eigenfunctions, and the lowest eigenvalue has an even eigenfunction. Thus $\mu(a)$ is the smallest strictly positive value of $\lambda^2$ solving~\eqref{eq:evaleqn}. Because the hyperbolic tangent is positive, and both terms are strictly increasing except at the singularities of the tangent, there is a unique solution to~\eqref{eq:evaleqn} in each interval 
on which the tangent runs over $(-\infty,0)$. The smallest solution is therefore is bounded so that
\begin{equation}
    \frac{\pi}{2}< \frac{\sqrt{\sqrt{a^{2}+\lambda^{2}}-a}}{2}
    <\pi
\end{equation}
or 
\begin{equation}
    (\pi^2+a)^2-a^2<\lambda^2< (4\pi^2+a)^2-a^2
\end{equation}
so $\lambda^2=\mathcal{O}(a)$ for large $a$. In this regime, the factor before the tangent is $\mathcal{O}(a^{-1/2})$ which means that the tangent should be large and negative, i.e., 
$\mu(a)\sim 2\pi^2 a$ for $a\gg \pi^2$. (The operator $\dd^4/\dd x^4$ is effectively acting as a perturbation of the operator $-2a\dd^2/\dd x^2$ with Dirichlet boundary conditions.) The solution for $a=0$ is $\mu(0)=500.5639018$ (thus $\mu(0)/(16\pi^2)=3.169857937$) and 
by the implicit function theorem we deduce that $\mu(a)$ is smooth near $a=0$ and in particular obeys $\mu(a)\to \mu(0)$ as $a\to 0$. Moreover $\mu(a)$ is monotone increasing in $a$, because $-\dd^2/\dd x^2$ is nonnegative. 

In~\cite{QEI2006FewsterPfenning}, section III.B, a similar method was used to analyse the scalar field on uniformly accelerated trajectories, using the specific dependence of $A$ and $B$ on the acceleration (and using the Fermi--Walker tetrad, which is equivalent to the Frenet--Serret in this case). There, however, the emphasis was on the magnitude and duration of periods in which the energy density drops below zero, rather than below the long-time limit $-B/(16\pi^2)$. 

\section{Proof of some lemmas}\label{appx:proof_of_lems}

We prove Lemmas~\ref{lem:trace_squared_C},~\ref{lem:K1},~\ref{lem:K}, and~\ref{lem:FofK}.

\subsection{$\Tr C(s,s')^2=(\Tr C(s,s'))^2$}
\label{appx:trace_squared_C}

The general line of argument in this subsection, and particularly the observation that $C$ is an improper Lorentz matrix, was suggested by an AI tool. However, the actual proof was generated and verified by the authors.
\begin{proof}[Proof of Lemma~\ref{lem:trace_squared_C}]
Recall that $C$ was defined by $C_{ab}(s,s') = \eta_{\mu\nu}e_a^{\mu}(s) e_b^\nu(s')$. As any two continuously related tetrads are related by a proper orthochronous Lorentz transformation, one has $e_a^\mu(s)=L(s,s')_a^{\phantom{a}c}e_c(s')$ for some smooth $L:\RR^2\to \Lor_{+}^\uparrow(1,3)$. Thus $C_{ab}(s,s') = L(s,s')_{a}^{\phantom{a}c}C_{cb}(s',s')= L(s,s')_{a}^{\phantom{a}c}\eta_{cb}$, which implies that each $C_{ab}(s,s')$ is the matrix of an improper orthochronous Lorentz transformation.

    As $C(s,s')$ is an improper Lorentz matrix, its algebraic eigenvalues are $1$, $-1$, $\lambda(s,s')$ and $\lambda(s,s')^{-1}$ for some $\lambda(s,s')\in\CC$ (with $\lambda(s,s)=-1$ for all $s$). The argument for this is essentially standard but we give it for completeness: any improper Lorentz matrix $L$ obeys $L^T\eta L=\eta$ and $\det L=-1$. Then $L^T\eta (L-\mu I)=\mu (L^T-\mu^{-1}I)\eta$ and so $\det(L-\mu I)=\mu^4 \det(L^T-\mu^{-1}I)/\det L^T=\mu^4 (\det L)\det(L-\mu^{-1}I)=-\mu^4 \det (L-\mu^{-1}I)$. Substituting $\mu=\pm 1$ yields a contradiction unless $\det (L-\mu I)=0$, so we deduce that $L$ has $\pm 1$ as eigenvalues. On the other hand, we also deduce that if $\mu$ is an eigenvalue then so is $\mu^{-1}$. Thus if $L$ has an eigenvalue other than $\pm 1$, it has four distinct eigenvalues $+1$, $-1$, $\mu$, $\mu^{-1}$. If all the roots of the characteristic polynomial are $\pm 1$, then either $1$ has algebraic multiplicity $1$ and $-1$ has algebraic multiplicity $3$, or vice versa, because the product of all algebraic eigenvalues is $\det L=-1$.
    Accordingly, the algebraic eigenvalues are again of the form
    $+1$, $-1$, $\mu$, $\mu^{-1}$, with $\mu\in\{\pm 1\}$.

    Summing the algebraic eigenvalues with multiplicity gives $\Tr C(s,s')=\lambda(s,s')+\lambda(s,s')^{-1}$.
    Meanwhile, $C(s,s')^2$ has algebraic eigenvalues $1$, $1$, $\lambda(s,s')^2$ and $\lambda(s,s')^{-2}$, so 
    $\Tr C(s,s')^2=\lambda(s,s')^2+2+\lambda(s,s')^{-2}=(\Tr C(s,s'))^2$.    
\end{proof} 

\subsection{Spectrum and (non)diagonalisability of $K$}  
    
    \begin{proof}[Proof of Lemma~\ref{lem:K1}]
    The statement is obvious for $\kappa=\tau=\upsilon=0$ because $K=0$, so we now assume $\kappa>0$. Direct calculation gives the characteristic polynomial of $K$ as
    \begin{equation}\label{eq:characteristic}
        \det(K-\lambda I) = \lambda^4 -(\kappa^2-\tau^2-\upsilon^2)\lambda^2 - \kappa^2\upsilon^2,
    \end{equation} 
    and consequently $K$ has eigenvalues
    \begin{equation}
       \lambda= \pm \frac{1}{\sqrt{2}}\left(
        \kappa^2-\tau^2-\upsilon^2 \pm R(\kappa,\tau,\upsilon)
        \right)^{1/2},
    \end{equation}
    where the two sign choices are independent and $R$ is as in~\eqref{eq:R}. 
    As $R\ge \kappa^2-\tau^2-\upsilon^2$, the eigenvalues can be written as 
    $\pm\chi$ and $\pm \ii\omega$ with $\chi,\omega\ge 0$ satisfying~\eqref{eq:chi} and~\eqref{eq:omega}.

    If $\upsilon\neq 0$ then none of the eigenvalues can vanish because $\det K=-\kappa^2\upsilon^2$; thus there are four distinct eigenvalues, so $K$ is diagonalisable.
    If $\upsilon=0$ then one of $\chi^2$, $\omega^2$ vanishes, while the other equals $|\kappa^2-\tau^2|$. 
    Provided $\kappa\neq \tau$, there are two distinct nonzero eigenvalues in addition to a zero eigenvalue. As $(\tau,0,-\kappa,0)^T$ and $(0,0,0,1)^T$ are linearly independent eigenvectors of zero eigenvalue, we may again conclude that $K$ is diagonalisable provided $\kappa\neq\tau$. 

    If $\kappa=\tau>\upsilon=0$ then the characteristic equation is $\det(K-\lambda I)=\lambda^4$ so the only eigenvalue is zero with algebraic multiplicity $4$. Direct calculation gives $K^2\neq 0$ but $K^3=0$, which shows that $K$ has a Jordan block of dimension $3$. The other block must have dimension $1$. The geometric multiplicity is therefore $2$.  
    \end{proof}
    
    \begin{proof}[Proof of Lemma~\ref{lem:K}] 
    Given the diagonal form $K=P\diag(\chi,-\chi,\ii\omega,-\ii\omega)P^{-1}$ as in~\eqref{eq:K_diagonal_form}, we have
    \begin{equation}
        K^2=P\diag(\chi^2,\chi^2,-\omega^2,-\omega^2)P^{-1}.
    \end{equation}  
    Thus $(\omega^2 I+K^2)/(\chi^2+\omega^2) = P\diag(1,1,0,0)P^{-1}$ is a projection onto the $\chi^2$-eigenspace of $K^2$ and $(\chi^2 I-K^2)/(\chi^2+\omega^2)=P\diag(0,0,1,1)P^{-1}$ projects onto the $-\omega^2$-eigenspace. This establishes~\eqref{eq:Pi}.
    Similarly, if $\chi>0$, one sees that $\Pi_+=P\diag(1,0,0,0)P^{-1}$ and $\Pi_-=P\diag(0,1,0,0)P^{-1}$, are rank-$1$ projections onto the $\pm\chi$-eigenspaces of $K$ and are equal to the formulae given in~\eqref{eq:Pi_pm}.

    The nilpotency of $K$ for $\kappa=\tau>0$, $\upsilon=0$ was mentioned in the proof of Lemma~\ref{lem:K1}. As all eigenvalues of $K$ vanish we have $\Tr K^r=0$ for all $r\in\NN$. 
    \end{proof}

\subsection{Functional calculus of $K$}

        \begin{proof}[Proof of Lemma~\ref{lem:FofK}]
         Parts~(b) and~(c) are immediate from~\eqref{eq:FofK} (and the fact that $\Tr K=\Tr K^2=0$ in these cases), so we need only address parts~(a),~(d) and~(e).
         
         For diagonalisable $K=P\Lambda P^{-1}$, $\Lambda=\diag(\chi,-\chi,\ii\omega,-\ii\omega)$ one has $F(sK)=PF(s\Lambda)P^{-1}$. If $F$ is even then
         \begin{equation}\label{eq:FofKeven}
             F(sK) = F(\chi s) \Pi + F(\ii\omega s) (I-\Pi),
         \end{equation}
         where $\Pi$ and $I-\Pi$ are projectors onto the $\chi^2$ and $-\omega^2$ eigenspaces of $K^2$ as in Lemma~\ref{lem:K}.  
         Noting that both $F_\textnormal{e}$ and $F_{\textnormal{oc}}$ are even analytic functions and $F(z) = F_\textnormal{e}(z)+zF_\textnormal{oc}(z)$, we have
         $F(sK)=F_\textnormal{e}(sK)+sKF_\textnormal{oc}(sK)$ using the power series definition. Using~\eqref{eq:FofKeven} to evaluate $F_\textnormal{e}(sK)$ and $F_\textnormal{oc}(sK)$ we obtain the required result~\eqref{eq:FofK_diag}. 

         Next, $K\Pi$ and $K(I-\Pi)$ are linear combinations of odd powers of $K$ and therefore have vanising trace. Thus one finds
         $\Tr F(s K) = F_\textnormal{e}(\chi s)\Tr \Pi + F_\textnormal{e}(\ii\omega s)\Tr(I-\Pi)$, giving the formula~\eqref{eq:TrFK}.

         Finally, if $G$ is another entire analytic function, we note that
         \begin{equation}
             (FG)_\textnormal{e}(z) = F_\textnormal{e}(z)G_\textnormal{e}(z)+z^2 F_\textnormal{oc}(z)G_\textnormal{oc}(z),
             \qquad
             (FG)_\textnormal{oc}(z) = F_\textnormal{e}(z)G_\textnormal{oc}(z)+F_\textnormal{oc}(z)G_\textnormal{e}(z)
         \end{equation}
        and the claim that $(FG)(sK)=F(sK)G(sK)$ follows from~\eqref{eq:FofK_diag} and the comparable formula for $G(sK)$, plus the properties of the projectors and the fact that $K^2\Pi = \chi^2$, $K^2(I-\Pi) = -\omega^2(I-\Pi)$.
        
         (d) This follows from the definition of $F(sK)$ as a power series, and the fact that
        \begin{equation}
        \eta K^T \eta^{-1} = (K\eta)^T\eta^{-1} = (-K\eta)\eta^{-1} = -K,
        \end{equation} 
        using the antisymmetry of $K_{\bullet\bullet}=K\eta$ mentioned earlier. 
     \end{proof} 

  For ease of reference, we give the explicit expressions
    \begin{equation}\label{eq:Piexplicit}
        \Pi = \frac{1}{\chi^2+\omega^2}\begin{pmatrix} 
            \kappa^2+\omega^2 & 0 & \kappa\tau & 0 \\ 
            0 & \kappa^2-\tau^2+\omega^2 & 0 & \tau\upsilon \\
            -\kappa\tau & 0 & -\tau^2-\upsilon^2+\omega^2 & 0 \\
            0 & \tau\upsilon & 0 & \omega^2-\upsilon^2
        \end{pmatrix},
    \end{equation}
    \begin{equation} \label{eq:KPi}
        K\Pi = \begin{pmatrix}
            0 & (\kappa^2+\omega^2-\tau^2)\kappa & 0 & \kappa\tau\upsilon \\
            (\kappa^2+\omega^2-\tau^2)\kappa & 0 & (\kappa^2+\omega^2-\tau^2-\upsilon^2)\tau & 0\\
            0 & -(\kappa^2+\omega^2-\tau^2-\upsilon^2)\tau & 0 & (\omega^2-\tau^2-\upsilon^2)\upsilon\\
            \kappa\tau\upsilon & 0 & (\tau^2+\upsilon^2-\omega^2)\upsilon & 0
        \end{pmatrix}.
    \end{equation}

    \subsection{Proofs of Lemmas~\ref{lem:D} and~\ref{lem:CD}}
    \label{appx:CD}
 
\begin{proof}[Proof of Lemma~\ref{lem:D}.]
    Using~\eqref{Dab_definition},~\eqref{eq:gamma_sol}, and 
    $e_a(s)=R(s)_a^{\phantom{a}b}e_b^\FS(s)=[R(s)\ee^{sK}]_a^{\phantom{a}b}e_b^\FS(0)$, 
    we compute
\begin{align}
    D_{ab}(s,s') &=-(s-s')^2 [\expc((s-s')K)\ee^{s'K}]_0^{\phantom{0}c}\eta_{cd}
    [R(s)\ee^{sK}]_a^{\phantom{a}d} 
    \nonumber \\ &\qquad\qquad\times
    [\expc((s-s')K)\ee^{s'K}]_0^{\phantom{0}g}\eta_{gh}[R(s')\ee^{s'K}]_b^{\phantom{b}h}. 
\end{align}   
Writing $v_0=(1,0,0,0)^T$, $D(s,s')$ may thus be expressed as a matrix product
\begin{align}\label{eq:Dssp}
  D(s,s')  &=-(s-s')^2 R(s)\ee^{sK}\eta (\expc((s-s')K)\ee^{s'K})^T v_0 v_0^T 
  \expc((s-s')K)\ee^{s'K}\eta (R(s')\ee^{s'K})^T  \nonumber\\
  &= -(s-s')^2R(s) \ee^{(s-s')K} \expc(-(s-s')K) \eta v_0 v_0^T  \expc((s-s')K)\eta R(s')^T\nonumber \\ 
  &=  
   -(s-s')^2  R(s) \expc((s-s')K)  v_0 v_0^T  \expc((s-s')K)\eta R(s')^T,
\end{align}
where we have used~\eqref{eq:functional_calc_id} and $\eta v_0=v_0$ together with 
$\exp(z)\expc(-z) = \expc(z)$. Taking the trace and using cyclicity,~\eqref{eq:TrDssp} is proved and we see that $\Tr D(s,s')$ is translationally invariant, even though $D(s,s')$ itself is not generally invariant. 

As
\begin{align}
    [\expc(-sK)\eta \ee^{s\Omega} \expc(-sK)]_{00} &=
    [\expc(-sK)\eta \ee^{s\Omega} \expc(-sK)\eta]_{00} \nonumber \\
    &= [\eta \expc(sK)^T (\ee^{-s\Omega})^T\eta \expc(sK)^T]_{00}
    \nonumber \\
    &=[\expc(sK)^T (\ee^{-s\Omega})^T\eta \expc(sK)^T]_{00}
    \nonumber \\&= [\expc(sK)\eta \ee^{-s\Omega}\expc(sK)]_{00},
\end{align}
we also see that $\Tr D(s,s')$ is even in $s-s'$, whereupon the existence of an entire analytic $\Hc_\Omega$ satisfying~\eqref{eq:Hdef} is established; moreover
$\Hc_\Omega$ is real on $\RR^+$ because $D(s,s')$ is real.  Symbolic calculation gives the Taylor expansion~\eqref{eq:H_Taylor}. 
 
In the case where $\kappa=0$,~\eqref{eq:TrD_kappa0} is established by a straightforward calculation  
\end{proof} 

\begin{proof}[Proof of Lemma~\ref{lem:CD}.]
Eq.~\eqref{eq:TrCDssp} is obtained from~\eqref{eq:Cssp} and~\eqref{eq:Dssp} by rotating the trace, thus establishing translational invariance. A similar argument to that used for $\Tr D$ shows that $\Tr C(s,s')D(s,s')$ is even in $s-s'$. Specifically, this follows from the calculation
\begin{align}
 [\expc(-sK)\eta
\ee^{s\Omega}
\ee^{-sK}\eta\ee^{s\Omega} \expc(-sK)]_{00} &=
 [\expc(-sK)\eta \ee^{s\Omega}\ee^{-sK}\eta\ee^{s\Omega} \expc(-sK)\eta]_{00} \nonumber\\
 &=[\eta\expc(sK)^T (\ee^{-s\Omega})^T\eta(\ee^{sK})^T(\ee^{-s\Omega})^T \eta \expc(sK)^T]_{00} \nonumber\\
  &=[\expc(sK)^T (\ee^{-s\Omega})^T\eta(\ee^{sK})^T(\ee^{-s\Omega})^T \eta \expc(sK)^T]_{00} \nonumber\\
  &= [\expc(sK)\eta \ee^{-s\Omega} \ee^{sK} \eta \ee^{-s\Omega}  \expc(sK)]_{00}.
\end{align}
Consequently,~\eqref{eq:Jdef} holds for an entire analytic function $\Jc_\Omega$ that is real-valued on $\RR^+$ because $C(s,s')$ and $D(s,s')$ are real matrices. The expansion~\eqref{eq:J_exp} follows by symbolic calculation and in the case $K=0$,~\eqref{eq:TrCDssp} immediately reduces to~\eqref{eq:TrCD_kappa0}.
\end{proof}

\section{Asymptotic behaviour} 
\label{appx:asymptotic} 

In this appendix, we prove Lemma~\ref{lem:K_asympt}, describing the asymptotic behaviour of the functions $K_\Omega(\kappa^2 s^2)$ and $K_{\textnormal{e.m.},\Omega}(\kappa^2 s^2)$ as $s\to\pm\infty$ in $\RR$. To do this, we need to understand the behaviour of $1/\Upsilon(\kappa^2s^2)$ in this limit, and also 
$\Tr C(s,s')$, $\Tr D(s,s')$ and $\Tr C(s,s')D(s,s')$ as 
$(s-s')\to \pm\infty$. 

The results divide into three cases depending on whether $K$ is (i) diagonalisable with $\chi=0$ (covering inertial and elliptic worldlines), (ii)  diagonalisable with $\chi>0$ (hyperbolic worldlines) or (iii) nondiagonalisable (parabolic worldlines). Recall that case  (ii) includes the linearly accelerated, catenary and loxodromic motions.  

Starting with $\Upsilon(\kappa^2 s^2)$, if $\chi>0$  then the $\sinhc$ term in~\eqref{eq:Upsilon_exact_a} dominates. Alternatively, if 
$\chi=0$ and $\omega>0$ (elliptic motion) then $\Upsilon(\kappa^2 s^2)$ tends to a nonzero constant; this is also true in the inertial case, in which $\Upsilon\equiv 1$. The behaviour in the parabolic case is immediate from~\eqref{eq:Upsilon_exact_b}. Overall, one has
\begin{equation}\label{eq:Upsilon_asymptotics}
    \frac{1}{\Upsilon(\kappa^2 s^2)} = 
    \begin{cases} 
    \mathcal{O}(1) & \textnormal{inertial, elliptic}\\    
    \mathcal{O}(s^2\ee^{-\chi |s|}) & \text{hyperbolic cases}\\
    \mathcal{O}(s^{-2}) & \textnormal{parabolic}
    \end{cases}
\end{equation}
as $|s|\to\infty$. 

Next, consider $ \exp(sK)$ as $s\to\pm\infty$ in $\RR$. 
Suppose that $K$ is diagonalisable and $\chi=0$. Then $K\Pi=0$ so~\eqref{eq:expsK} gives 
\begin{equation}
    \exp(sK) =\Pi +  (\cos(\omega s)+s\sinc(\omega s)K)(I-\Pi),
\end{equation}
reducing to $\exp(sK)=I$ if $K=0$.  
Thus in either case we have $\exp(sK)=\mathcal{O}(1)$ as $|s|\to\infty$. Alternatively, if $K$ is diagonalisable with $\chi>0$ then~\eqref{eq:expsK} immediately implies that
\begin{equation}
    \exp(sK) = \frac{\ee^{\chi|s|}}{2\chi}(\chi I\pm K)\Pi +\mathcal{O}(1) = \ee^{\chi|s|}\Pi_\pm +\mathcal{O}(1)
\end{equation}
as $|s|\to\infty$, where we recall from  Lemma~\ref{lem:K} that $\Pi_\pm$ is the rank-$1$ projection (not necessarily orthogonal) onto the $K$-eigenspaces of eigenvalue $\pm\chi$. 
Finally, in the parabolic case $\exp(sK)$ is given by~\eqref{eq:expsK_nondiag}. 
Altogether, we have 
    \begin{equation}
        \exp(sK) = \begin{cases} 
        \mathcal{O}(1) & \text{inertial, elliptic}\\
        \ee^{\chi |s|} \Pi_\pm + \mathcal{O}(1) & \text{hyperbolic cases}\\
        \tfrac{1}{2}s^2K^2 + sK + I & \text{parabolic}.
        \end{cases}
    \end{equation}
as $s\to\pm\infty$.  
Consequently, as $\exp(-(s-s')\Omega)$ is bounded, we may read off from~\eqref{eq:TrC_simp} that
\begin{equation}
    \Tr C(s,s') = \begin{cases} 
        \mathcal{O}(1) & \text{inertial, elliptic}\\
        \ee^{\chi |s-s'|} \Tr (\eta \ee^{-(s-s')\Omega}\Pi_\pm) + \mathcal{O}(1) & \text{hyperbolic cases}\\
        \tfrac{1}{2}(s-s')^2 \Tr (\eta  \ee^{-(s-s')\Omega}K^2)+ \mathcal{O}(s-s') & \text{parabolic}.
        \end{cases}
    \end{equation}
as $s\to\pm\infty$.  

Turning to the asymptotics of $\expc(sK)$,~\eqref{eq:fsK} implies that 
\begin{equation}
    \expc(sK) = \frac{\ee^{\chi |s|}}{2\chi^2 |s|}(\chi I \pm K)\Pi + \mathcal{O}(1) = 
\frac{\ee^{\chi |s|}}{\chi |s|}\Pi_{\pm} + \mathcal{O}(1)
\end{equation}
as $s\to\pm\infty$ in the case where $K$ is diagonalisable and $\chi>0$. Adding in the other cases, one has
\begin{equation}\label{eq:fsKasympt}
    \expc(sK) = \begin{cases}
    \mathcal{O}(1) & \text{inertial, elliptic}\\
    \frac{\ee^{\chi |s|}}{\chi |s|}\Pi_{\pm} + \mathcal{O}(1) & \text{hyperbolic cases}\\ 
         s^2 K^2/6 + s K/2+I & \text{parabolic}
    \end{cases}
\end{equation}
as $s\to\pm\infty$. Here, the elliptic case follows because $\expc(\pm\ii\omega s)$ are bounded as $|s|\to\infty$ -- this can also be seen from~\eqref{eq:fsK} noting that $K\Pi=0$ if $\chi=0$.

Using~\eqref{eq:TrDssp}, we can now read off the asymptotic behaviour of $\Tr D$ as
\begin{equation}\label{eq:TrDasympt}
    \Tr D(s,s') = \begin{cases}
    \mathcal{O}((s-s')^2) & \text{inertial, elliptic}\\
    -\ee^{2\chi |s-s'|}[\Pi_{\pm}\eta\ee^{-(s-s')\Omega}\Pi_{\pm}]_{00}/\chi^2 + \mathcal{O}((s-s')\ee^{\chi |s-s'|}) & \text{hyperbolic cases}\\ 
         -(s-s')^6 [K^2\eta \ee^{-(s-s')\Omega} K^2]_{00}/36 +\mathcal{O}((s-s')^5) & \text{parabolic.}
    \end{cases}
\end{equation}
There is a useful simplification in the hyperbolic cases because $\Pi_\pm$ has rank-$1$, leading to the factorisation
\begin{equation}
    [\Pi_{\pm}\eta\ee^{-(s-s')\Omega}\Pi_{\pm}]_{00} =
    [\Pi_\pm]_{00} \Tr (\eta \ee^{-(s-s')\Omega}\Pi_\pm).
\end{equation}
 
Finally, we use~\eqref{eq:TrCDssp} to obtain  
\begin{multline}\label{eq:TrCDasympt}
    \Tr C(s,s')D(s,s') = \\
    \begin{cases}
    \mathcal{O}((s-s')^2) & \text{inertial, elliptic}\\
    -\ee^{3\chi |s-s'|}\chi^{-2}[\Pi_{\pm}\eta\ee^{-(s-s')\Omega}\Pi_{\pm}\eta\ee^{-(s-s')\Omega}\Pi_{\pm}]_{00} + \mathcal{O}((s-s')\ee^{2\chi |s-s'|}) & \text{hyperbolic}\\ 
         -(s-s')^8 [K^2\eta \ee^{-(s-s')\Omega}K^2\eta\ee^{-(s-s')\Omega} K^2]_{00}/72 +\mathcal{O}((s-s')^7) & \text{parabolic}
    \end{cases}
\end{multline} 
Again, because $\Pi_\pm$ have rank-$1$, there is a factorisation
\begin{equation}
    [\Pi_{\pm}\eta\ee^{-(s-s')\Omega}\Pi_{\pm}\eta\ee^{-(s-s')\Omega}\Pi_{\pm}]_{00} = 
    [\Pi_\pm]_{00}(\Tr \eta \ee^{-(s-s')\Omega}\Pi_{\pm})^2
\end{equation}
which entails that $\Tr C \Tr D$ and $\Tr CD$ have the same leading term in the hyperbolic situation.
In the parabolic situation, there is a similar factorisation 
\begin{equation}
    [K^2\eta \ee^{-(s-s')\Omega}K^2\eta\ee^{-(s-s')\Omega} K^2]_{00} = 
     [K^2\eta \ee^{-(s-s')\Omega} K^2]_{00} \Tr (\eta  \ee^{-(s-s')\Omega}K^2)
\end{equation}  
because $K^2$ has rank-$1$ in that case (specifically, $K^2 =\kappa^2 (1,0,-1,0)^T(1,0,1,0)$).
Thus the leading terms of  $\Tr C \Tr D$ and $\Tr CD$  coincide again. There is no such cancellation in the inertial or elliptic cases (see, for example~\eqref{eq:TrCTrDminusTrCD_kappa0}). 
In summary, we find that  
\begin{equation}
    \Tr C(s,s')\Tr D(s,s')-\Tr C(s,s')D(s,s') = \begin{cases}  
    \mathcal{O}((s-s')^2) & \text{inertial, elliptic}\\ 
    \mathcal{O}((s-s')\ee^{2\chi |s-s'|}) & \text{hyperbolic}\\ 
         \mathcal{O}((s-s')^7) & \text{parabolic}
    \end{cases}
\end{equation}
as $|s|\to\infty$. 

Finally, the results obtained above can be combined straightforwardly to find the asymptotics of 
\begin{equation} 
    \Kc_{\Omega}(\kappa^2s^2)= \frac{\Gc_\Omega(\kappa^2s^2)\Upsilon(\kappa^2s^2)-4\Hc(\kappa^2s^2)}{\Upsilon(\kappa^2s^2)^3} 
\end{equation}
and  
\begin{equation}
    \Kc_{\textnormal{e.m.},\Omega}(\kappa^2s^2)= -4\frac{\Xi_\Omega(\kappa^2s^2)}{\Upsilon(\kappa^2s^2)^3}, 
\end{equation}
thus proving Lemma~\ref{lem:K_asympt}. The cancellation between the leading terms of $\Tr CD$ and $(\Tr C)(\Tr D)$ is crucial in the hyperbolic case to ensure that $\Kc_{\textnormal{e.m.},\Omega}$ decays rather than growing exponentially.

\section{Optimisation of the scalar QEI}\label{appx:optimisation}

In this appendix we give details on the approximations given in Section~\ref{sec:optimisation} and the numerical method used to generate Fig.~\ref{fig:min}.

The problem is to minimise $B_\Omega$ (see~\eqref{eq:B_scalar_Omega}) over $\boldsymbol{\varpi}$ with fixed curvature scalars. We may exclude the case $\kappa=\tau=\upsilon=0$ for which $B_\Omega$ vanishes identically. 
Noting that
\begin{equation}
  \frac{\partial B_\Omega}{\partial\varpi_2}= \frac{2\varpi_2}{3}\left(
  5 \kappa^{2}+\tau^{2}+\upsilon^{2}+2 \varpi_{2}^{2}+2 \left(\varpi_{1}-\upsilon \right)^{2}+2 \left(\varpi_{3}-\tau \right)^{2}
\right),
\end{equation}
it is then clear that critical points all have $\varpi_2=0$, so we may restrict to this situation. We also write $\tau=t\kappa$, $\upsilon=u\kappa$ and $\varpi_1=\upsilon+\zeta\kappa$, $\varpi_3=\tau+\theta\kappa$, allowing us to write $B_\Omega$ in terms of the curvature $\kappa$ and dimensionless parameters $t\ge 0$, $u\in\RR$ and $\zeta,\theta\in\RR$ to give
\begin{equation} 
    B_\Omega = \frac{\kappa^4}{90}\left(
    33+26 t^{2}+30q(\zeta,\theta)\right),
\end{equation} 
where the function $q$, which depends on $t$ and $u$ as parameters, is 
\begin{equation} 
   q(\zeta,\theta) = (\theta^{2}+ \zeta^{2})^2+ 2 \zeta^{2}
   +5(\theta+t/5)^2 +(u\theta-t\zeta)^2. 
\end{equation}
In these variables, the Fermi--Walker tetrad corresponds to $\zeta=\theta=0$, and 
\begin{equation}
\beta=26+30 \min_{\zeta,\theta\in\RR}\frac{q(\zeta,\theta)}{t^2}
\end{equation}
clearly depends only on $t$ and $u$.
As a sum of squares, it is clear that $q\ge 0$, and since the
squares cannot all vanish simultaneously if $t>0$ we see that $q>0$ for $t>0$. Thus
$\beta>26$ if $t>0$, which improves on the lower bound $\beta\ge 22$ given by~\eqref{eq:B_bounds}. 
If $u=0$ then $q(0,-t/5)/t^2=t^2/625\to 0$ as $t\to 0$, and we deduce that  
$\beta$ has infimum $26$ over all $t>0$, $u,\zeta,\theta\in \RR$. 

We wish to find the global minimum of $q$ for given $t,u$. Straightforward calculation gives the trace and determinant of the Hessian $\Hess(q)$ as
\begin{align}
    \Tr \Hess(q)&=16(\theta^2+\zeta^2)+2(t^2+u^2)+14 \\
    \det \Hess(q)&= 48 \left(\theta^{2}+\zeta^{2}\right)^{2}+16 \left(t \theta +u \zeta \right)^{2}+8 \left(t^{2}+u^{2}+11\right) \theta^{2}+8 \left(t^{2}+u^{2}+17\right) \zeta^{2}
    \nonumber \\
    &\qquad
    +20 t^{2}+8 u^{2}+40,
\end{align}
which are both manifestly positive for all parameter values, so $\Hess(q)$ is everywhere positive definite and $q$ is therefore strictly convex, with a unique local minimum that is also the global minimum. Therefore the minimum can be found for given $(t,u)$ by solving for a critical point $(\zeta_*,\theta_*)$ of $q$ and substituting back into $q$. As $q(\zeta_*,\theta_*)\le t^2/5$ (because $\beta\le 32$) we have $\zeta_*^2+\theta_*^2\le t/\sqrt{5}$, so the search space is bounded for given $t$. This method was used to generate Fig.~\ref{fig:min}, which indicates that $\beta\to 26$ as $(t,u)\to (0,0)$ and the Fermi--Walker value $32$ approached for any $t>0$ as $|u|\to\infty$.

The numerics can be supported by the following analytical observations. The equations for a critical point are
\begin{align}\label{eq:dqdzeta}
    0=\frac{1}{2}\frac{\partial q}{\partial\zeta} &= 2\zeta(\theta^2+\zeta^2)-tu\theta+(t^2+2)\zeta  =2\zeta(1+\theta^2+\zeta^2) -t(u\theta-t\zeta)
    \\
    \label{eq:dqdtheta}
    0=\frac{1}{2}\frac{\partial q}{\partial\theta} &=
    2\theta(\theta^2+\zeta^2)-tu\zeta+(u^2+5)\theta  +t = t + 2\theta\left(\frac{5}{2}+\theta^2+\zeta^2\right)   +u (u\theta-t\zeta)
    .
\end{align}
Substituting $\theta=t \alpha$, $\zeta=ut^2 \beta$, and dividing out factors of $t$ or $ut^2$, these equations become (in reverse order)
\begin{align}
    1+5\alpha + \alpha u^2 +2 t^{2} \alpha^{3} -t^{2} u^{2} \beta +
2 \alpha  \beta^{2} t^{4} u^{2}  &= 0 \\
2\beta - \alpha +(2\alpha^2+1)\beta t^2 + 2t^4u^2\beta^3   &=0.
\end{align}
These equations can be solved by power series expansion methods to give
$\alpha= -1/5 + \mathcal{O}((t^2+u^2))$, $\beta= -1/10 + \mathcal{O}((t^2+u^2))$, 
and consequently 
\begin{align}
    \theta_* &=-\frac{1}{5} t\left(1+\mathcal{O}(t^2+u^2)\right), \qquad  \zeta_*=-\frac{1}{10} ut^2 \left(1+\mathcal{O}(t^2+u^2)\right),
\end{align}
as $t,u\to 0$, corresponding to the values in~\eqref{eq:optvarpi_small}, and the approximations for $\beta$ and $B_\textnormal{min}$ given in~\eqref{eq:beta_small} and~\eqref{eq:Bmin_small}. For example, parabolic stationary worldlines correspond to $t=1$, $u=0$. Here, we find the numeric values
\begin{equation}
    \theta_* = -0.1969444375,\qquad \zeta_*=0,\qquad \beta(1,0)= 26.04653366,
\end{equation}
corresponding to the values given after~\eqref{eq:optvarpi_small}. 

Now consider limits  $t^2+u^2\to\infty$ with $t/|u|>\epsilon$ for some $\epsilon>0$. 
Equivalently, we write $t=r\cos\alpha$, $u=r\sin\alpha$ with $|\alpha|<\pi/2$ and consider $r\to\infty$ with $\cos\alpha$ bounded away from zero. Writing $\theta_*=-\rho\cos\beta$, $\zeta_* = -\rho\sin\beta$, the combination
$u\times\eqref{eq:dqdzeta}+t\times\eqref{eq:dqdtheta}$ gives
\begin{equation}\label{eq:firstcomb}
   r\cos^2\alpha - 2\rho^3\cos(\beta-\alpha) - \frac{\rho}{2} \left( 7\cos(\beta-\alpha) + 3\cos(\alpha+\beta)\right)  = 0
\end{equation}
on dividing by $r$. It follows that $|\rho|\to\infty$. On the other hand, 
$u\times \eqref{eq:dqdtheta}- t\times\eqref{eq:dqdzeta}$ gives
\begin{equation}\label{eq:secondcomb}
   \left(1+2\frac{\rho^2}{r^2}\right) \sin(\beta-\alpha)+ \frac{\sin(2\alpha)}{2\rho r} -
   \frac{1}{2r^2}\left(3\sin(\alpha+\beta) + 7\sin(\alpha-\beta)\right)
   =0 
\end{equation}
on dividing by $r^3\rho$,
from which we find that $\alpha-\beta\to 0$ and so $\cos(\beta-\alpha)\to 1$. Going back to~\eqref{eq:firstcomb} we find $\rho=\mathcal{O}(r^{1/3})$ as $r\to\infty$. Thus both $\theta$ and $\zeta$ are $\mathcal{O}(r^{1/3})$, while $r=t/\cos\alpha=\mathcal{O}(t)$. From~\eqref{eq:secondcomb} we now deduce that $\sin(\alpha-\beta)=\mathcal{O}(t^{-4/3})$.
In this limit, 
\begin{equation}
    q(\zeta_*,\theta_*) = \frac{t^2}{5} + \mathcal{O}(t^{4/3}) 
\end{equation}
and hence $\beta(t,u)= 32+\mathcal{O}(t^{-1/3})$ in such limits. In fact, $ \rho \sim \left(r\cos^2\alpha/2\right)^{1/3}$, so
\begin{align} 
\theta_* &\sim -\cos\alpha \left(\frac{r\cos^2\alpha}{2}\right)^{1/3}= -\frac{2^{-1/3} t^{5/3}}{(u^2+t^2)^{2/3}}, \nonumber \\  \zeta_* &\sim -\sin\alpha\left(\frac{r\cos^2\alpha}{2}\right)^{1/3} = -\frac{2^{-1/3}t^{2/3}u}{(u^2+t^2)^{2/3}}, 
\end{align}
giving
\begin{equation}
    \varpi_1 \sim \upsilon\left(1 -\frac{(\kappa\tau/\sqrt{2})^{2/3}}{(\tau^2+\upsilon^2)^{2/3}} \right) , \qquad
    \varpi_3 \sim \tau\left(1 -\frac{(\kappa\tau/\sqrt{2})^{2/3}}{(\tau^2+\upsilon^2)^{2/3}}\right),
\end{equation}
which shows an asymptotic approach to the Fermi--Walker tetrad.

Other limits can be studied. For example, if $|u|\to\infty$ with $t$ bounded (and hence $\zeta_*,\theta_*$ bounded), then~\eqref{eq:dqdtheta} implies that $u\theta_*-t\zeta_*=\mathcal{O}(u^{-1})$, whereupon~\eqref{eq:dqdzeta} gives $\zeta_*=\mathcal{O}(u^{-1})$ and therefore $\theta_*=t\zeta_*/u=\mathcal{O}(u^{-2})$. It follows that
\begin{equation}
    u\theta_*-t\zeta_* \sim -\frac{t}{u}, \qquad \zeta_*\sim -\frac{t^2}{2u}, \qquad \theta_* \sim -\frac{t(t^2+2)}{2u^2},
\end{equation}
and we again find $\beta(t,u)\to 32$ and $\varpi_1\sim\upsilon$, $\varpi_3\sim \tau$ (assuming $\tau/\kappa$ is kept bounded away from zero).

% 
%\bibliographystyle{spmpsci_mod2019}
%{\small
%\bibliography{references1}

\begin{thebibliography}{10}\setlength{\itemsep}{-1.5mm}
	\providecommand{\url}[1]{{#1}}
	\providecommand{\urlprefix}{URL }
	\expandafter\ifx\csname urlstyle\endcsname\relax
	\providecommand{\doi}[1]{DOI~\discretionary{}{}{}#1}\else
	\providecommand{\doi}{DOI~\discretionary{}{}{}\begingroup
		\urlstyle{rm}\Url}\fi
	
	\bibitem{Alcubierre_1994}
	Alcubierre, M.: The warp drive: hyper-fast travel within general relativity.
	\newblock Classical and Quantum Gravity \textbf{11}(5), L73--L77 (1994).
	\newblock \doi{10.1088/0264-9381/11/5/001}.
	\newblock \urlprefix\url{https://doi.org/10.1088%2F0264-9381%2F11%2F5%2F001}
	
	\bibitem{BellLeinaas:1983}
	Bell, J., Leinaas, J.: Electrons as accelerated thermometers.
	\newblock Nuclear Physics B \textbf{212}(1), 131--150 (1983).
	\newblock \doi{https://doi.org/10.1016/0550-3213(83)90601-6}.
	\newblock
	\urlprefix\url{https://www.sciencedirect.com/science/article/pii/0550321383906016}
	
	\bibitem{Bunney:2024}
	Bunney, C.R.D.: Stationary trajectories in {M}inkowski spacetimes.
	\newblock J. Math. Phys. \textbf{65}(5), Paper No. 052501, 14 (2024).
	\newblock \doi{10.1063/5.0205471}.
	\newblock \urlprefix\url{https://doi.org/10.1063/5.0205471}
	
	\bibitem{CandelasDeutsch:1977}
	Candelas, P., Deutsch, D.: On the vacuum stress induced by uniform acceleration
	or supporting the ether.
	\newblock Proc. Roy. Soc. London Ser. A \textbf{354}(1676), 79--99 (1977).
	\newblock \doi{10.1098/rspa.1977.0057}.
	\newblock \urlprefix\url{https://doi.org/10.1098/rspa.1977.0057}
	
	\bibitem{DawsonFewster:2006}
	Dawson, S.P., Fewster, C.J.: An explicit quantum weak energy inequality for
	{D}irac fields in curved spacetimes.
	\newblock Classical Quantum Gravity \textbf{23}(23), 6659--6681 (2006).
	\newblock \doi{10.1088/0264-9381/23/23/005}.
	\newblock \urlprefix\url{https://doi-org/10.1088/0264-9381/23/23/005}
	
	\bibitem{DeBievreMerkli:2006}
	De~Bi\`evre, S., Merkli, M.: The {U}nruh effect revisited.
	\newblock Classical Quantum Gravity \textbf{23}(22), 6525--6541 (2006).
	\newblock \doi{10.1088/0264-9381/23/22/026}.
	\newblock \urlprefix\url{https://doi.org/10.1088/0264-9381/23/22/026}
	
	\bibitem{EpsteinGlaserJaffe:1965}
	Epstein, H., Glaser, V., Jaffe, A.: Nonpositivity of the energy density in
	quantized field theories.
	\newblock Nuovo Cimento (10) \textbf{36}, 1016--1022 (1965).
	\newblock \doi{10.1007/BF02749799}.
	\newblock \urlprefix\url{https://doi.org/10.1007/BF02749799}
	
	\bibitem{FewsterEveson:2007}
	Eveson, S.P., Fewster, C.J.: Mass dependence of quantum energy inequality
	bounds.
	\newblock J. Math. Phys. \textbf{48}(9), 093506, 13 (2007).
	\newblock \doi{10.1063/1.2779137}.
	\newblock \urlprefix\url{https://doi.org/10.1063/1.2779137}
	
	\bibitem{GenQEI2000CFewster}
	Fewster, C.J.: A general worldline quantum inequality.
	\newblock Classical and Quantum Gravity \textbf{17}(9), 1897–1911 (2000).
	\newblock \doi{10.1088/0264-9381/17/9/302}.
	\newblock \urlprefix\url{http://dx.doi.org/10.1088/0264-9381/17/9/302}
	
	\bibitem{Fews05}
	Fewster, C.J.: Energy inequalities in quantum field theory.
	\newblock In: J.C. Zambrini (ed.) XIVth International Congress on Mathematical
	Physics. World Scientific, Singapore (2005).
	\newblock \doi{10.48550/arXiv.math-ph/0501073}.
	\newblock \urlprefix\url{https://doi.org/10.48550/arXiv.math-ph/0501073}.
	\newblock An expanded and updated version is available as {\tt math-ph/0501073}
	
	\bibitem{Fewster2017}
	Fewster, C.J.: Quantum energy inequalities.
	\newblock In: F.S.N. Lobo (ed.) Wormholes, Warp Drives and Energy Conditions,
	\emph{Fundamental Theories of Physics}, vol. 189, pp. 215--254. Springer
	International Publishing (2017).
	\newblock \doi{10.1007/978-3-319-55182-1_10}.
	\newblock \urlprefix\url{https://doi.org/10.1007/978-3-319-55182-1_10}
	
	\bibitem{BoundNegEnDenFewsterEveson1998}
	Fewster, C.J., Eveson, S.P.: Bounds on negative energy densities in flat
	spacetime.
	\newblock Phys. Rev. D (3) \textbf{58}(8), {084}{010}, 6 (1998).
	\newblock \doi{10.1103/PhysRevD.58.084010}.
	\newblock \urlprefix\url{https://doi.org/10.1103/PhysRevD.58.084010}
	
	\bibitem{FewsterKontou2022}
	Fewster, C.J., Kontou, E.A.: A semiclassical singularity theorem.
	\newblock Classical Quantum Gravity \textbf{39}(7), Paper No. 075028, 32
	(2022).
	\newblock \doi{10.1088/1361-6382/ac566b}.
	\newblock \urlprefix\url{https://doi.org/10.1088/1361-6382/ac566b}
	
	\bibitem{FewsterMistry:2003}
	Fewster, C.J., Mistry, B.: Quantum weak energy inequalities for the {D}irac
	field in flat spacetime.
	\newblock Phys. Rev. D (3) \textbf{68}(10), 105010, 6 (2003).
	\newblock \doi{10.1103/PhysRevD.68.105010}.
	\newblock \urlprefix\url{https://doi-org/10.1103/PhysRevD.68.105010}
	
	\bibitem{QEI2006FewsterPfenning}
	Fewster, C.J., Pfenning, M.: Quantum energy inequalities and local covariance.
	{I}. {G}lobally hyperbolic spacetimes.
	\newblock Journal of Mathematical Physics \textbf{47}(8), 082303 (2006).
	\newblock \doi{10.1063/1.2212669}.
	\newblock \urlprefix\url{http://dx.doi.org/10.1063/1.2212669}
	
	\bibitem{FewsterPfenning:2003}
	Fewster, C.J., Pfenning, M.J.: A quantum weak energy inequality for spin-one
	fields in curved space-time.
	\newblock J. Math. Phys. \textbf{44}(10), 4480--4513 (2003).
	\newblock \doi{10.1063/1.1602554}.
	\newblock \urlprefix\url{https://doi.org/10.1063/1.1602554}
	
	\bibitem{FewsterRoman:2005}
	Fewster, C.J., Roman, T.A.: On wormholes with arbitrarily small quantities of
	exotic matter.
	\newblock Phys. Rev. D (3) \textbf{72}(4), 044023, 15 (2005).
	\newblock \doi{10.1103/PhysRevD.72.044023}.
	\newblock \urlprefix\url{https://doi.org/10.1103/PhysRevD.72.044023}
	
	\bibitem{FewsterTeo:1999}
	Fewster, C.J., Teo, E.: Bounds on negative energy densities in static
	space-times.
	\newblock Phys. Rev. D \textbf{59}, 104016 (1999).
	\newblock \doi{10.1103/PhysRevD.59.104016}.
	\newblock \urlprefix\url{http://dx.doi.org/10.1103/PhysRevD.59.104016}
	
	\bibitem{FewsterTeo:2000}
	Fewster, C.J., Teo, E.: Quantum inequalities and ``quantum interest'' as
	eigenvalue problems.
	\newblock Phys. Rev. D (3) \textbf{61}(8), 084012, 14 (2000).
	\newblock \doi{10.1103/PhysRevD.61.084012}.
	\newblock \urlprefix\url{https://link.aps.org/doi/10.1103/PhysRevD.61.084012}
	
	\bibitem{FewsterThompson2023}
	Fewster, C.J., Thompson, J.: Quantum energy inequalities along stationary
	worldlines.
	\newblock Classical Quantum Gravity \textbf{40}(17), Paper No. 175008, 32
	(2023).
	\newblock \doi{10.1088/1361-6382/ace233}.
	\newblock \urlprefix\url{https://doi.org/10.1088/1361-6382/ace233}
	
	\bibitem{FewsterVerch:2001}
	Fewster, C.J., Verch, R.: A quantum weak energy inequality for {D}irac fields
	in curved spacetime.
	\newblock Comm. Math. Phys. \textbf{225}(2), 331--359 (2002).
	\newblock \doi{10.1007/s002200100584}.
	\newblock \urlprefix\url{https://doi-org/10.1007/s002200100584}
	
	\bibitem{FewsterVerch:2003}
	Fewster, C.J., Verch, R.: Stability of quantum systems at three scales:
	passivity, quantum weak energy inequalities and the microlocal spectrum
	condition.
	\newblock Comm. Math. Phys. \textbf{240}(1-2), 329--375 (2003).
	\newblock \doi{10.1007/s00220-003-0884-7}.
	\newblock \urlprefix\url{https://doi-org/10.1007/s00220-003-0884-7}
	
	\bibitem{Ford1978QuantumCE}
	Ford, L.H.: Quantum coherence effects and the second law of thermodynamics.
	\newblock Proceedings of the Royal Society of London. A. Mathematical and
	Physical Sciences \textbf{364}, 227 -- 236 (1978).
	\newblock \doi{10.1098/rspa.1978.0197}.
	\newblock \urlprefix\url{https://doi.org/10.1098/rspa.1978.0197}
	
	\bibitem{Ford:1991}
	Ford, L.H.: Constraints on negative-energy fluxes.
	\newblock Phys. Rev. D \textbf{43}, 3972--3978 (1991).
	\newblock \doi{10.1103/PhysRevD.43.3972}.
	\newblock \urlprefix\url{https://link.aps.org/doi/10.1103/PhysRevD.43.3972}
	
	\bibitem{FordRoman:1995}
	Ford, L.H., Roman, T.A.: Averaged energy conditions and quantum inequalities.
	\newblock Phys. Rev. D (3) \textbf{51}(8), 4277--4286 (1995).
	\newblock \doi{10.1103/PhysRevD.51.4277}.
	\newblock \urlprefix\url{https://doi.org/10.1103/PhysRevD.51.4277}
	
	\bibitem{FordRoman:1996worm}
	Ford, L.H., Roman, T.A.: Quantum field theory constrains traversable wormhole
	geometries.
	\newblock Phys. Rev. D (3) \textbf{53}(10), 5496--5507 (1996).
	\newblock \doi{10.1103/PhysRevD.53.5496}.
	\newblock \urlprefix\url{https://doi.org/10.1103/PhysRevD.53.5496}
	
	\bibitem{GoodJuarezAubryetal:2020}
	Good, M., Ju\'{a}rez-Aubry, B.A., Moustos, D., Temirkhan, M.: Unruh-like
	effects: effective temperatures along stationary worldlines.
	\newblock J. High Energy Phys. \textbf{2020}(6), 059, 25 (2020).
	\newblock \doi{10.1007/jhep06(2020)059}.
	\newblock \urlprefix\url{https://doi.org/10.1007/jhep06(2020)059}
	
	\bibitem{Gooding_etal_PRL:2020}
	Gooding, C., Biermann, S., Erne, S., Louko, J., Unruh, W.G., Schmiedmayer, J.,
	Weinfurtner, S.: Interferometric {U}nruh detectors for {B}ose-{E}instein
	condensates.
	\newblock Phys. Rev. Lett. \textbf{125}(21), 213603, 6 (2020).
	\newblock \doi{10.1103/physrevlett.125.213603}.
	\newblock
	\urlprefix\url{https://link.aps.org/doi/10.1103/PhysRevLett.125.213603}
	
	\bibitem{Hawking:1966}
	Hawking, S.W.: The occurrence of singularities in cosmology. {I}.
	\newblock Proc. Roy. Soc. London Ser. A \textbf{294}, 511--521 (1966).
	\newblock \doi{10.1098/rspa.1966.0221}.
	\newblock \urlprefix\url{https://doi.org/10.1098/rspa.1966.0221}
	
	\bibitem{KontouSanders:2020}
	Kontou, E.A., Sanders, K.: Energy conditions in general relativity and quantum
	field theory.
	\newblock Classical Quantum Gravity \textbf{37}(19), 193001, 63 (2020).
	\newblock \doi{10.1088/1361-6382/ab8fcf}.
	\newblock \urlprefix\url{https://dx.doi.org/10.1088/1361-6382/ab8fcf}
	
	\bibitem{Kottler1912}
	Kottler, F.: Über die {R}aumzeitlinien der {M}inkowski'schen {W}elt.
	\newblock Sitzungsberichte der mathematisch-naturwissenschaftlichen Classe der
	Kaiserlichen Akademie der Wissenschaften, Wien \textbf{121, 2a}, 1659--1759
	(1912).
	\newblock
	\urlprefix\url{https://viewer.acdh.oeaw.ac.at/viewer/image/MN_2Abt_a_121_2_1912/238/}
	
	\bibitem{Letaw:1981}
	Letaw, J.R.: Stationary world lines and the vacuum excitation of noninertial
	detectors.
	\newblock Phys. Rev. D \textbf{23}, 1709--1714 (1981).
	\newblock \doi{10.1103/PhysRevD.23.1709}.
	\newblock \urlprefix\url{https://link.aps.org/doi/10.1103/PhysRevD.23.1709}
	
	\bibitem{LetawPfautsch:1981}
	Letaw, J.R., Pfautsch, J.D.: Quantized scalar field in the stationary
	coordinate systems of flat spacetime.
	\newblock Phys. Rev. D (3) \textbf{24}(6), 1491--1498 (1981).
	\newblock \doi{10.1103/PhysRevD.24.1491}.
	\newblock \urlprefix\url{https://doi.org/10.1103/PhysRevD.24.1491}
	
	\bibitem{LetawPfautsch:1982}
	Letaw, J.R., Pfautsch, J.D.: The stationary coordinate systems in flat
	spacetime.
	\newblock J. Math. Phys. \textbf{23}(3), 425--431 (1982).
	\newblock \doi{10.1063/1.525364}.
	\newblock \urlprefix\url{https://doi.org/10.1063/1.525364}
	
	\bibitem{misner2017gravitation}
	Misner, C., Thorne, K., Wheeler, J., Kaiser, D.: Gravitation.
	\newblock Princeton University Press (2017)
	
	\bibitem{MorettiVanzo:1996}
	Moretti, V., Vanzo, L.: Thermal {W}ightman functions and renormalized stress
	tensors in the {R}indler wedge.
	\newblock Phys. Lett. B \textbf{375}(1-4), 54--59 (1996).
	\newblock \doi{10.1016/0370-2693(96)00223-7}.
	\newblock \urlprefix\url{https://doi.org/10.1016/0370-2693(96)00223-7}
	
	\bibitem{MorrisThorne:1988}
	Morris, M.S., Thorne, K.S.: Wormholes in spacetime and their use for
	interstellar travel: a tool for teaching general relativity.
	\newblock Amer. J. Phys. \textbf{56}(5), 395--412 (1988).
	\newblock \doi{10.1119/1.15620}.
	\newblock \urlprefix\url{https://doi.org/10.1119/1.15620}
	
	\bibitem{Penrose:1965}
	Penrose, R.: Gravitational collapse and space-time singularities.
	\newblock Phys. Rev. Lett. \textbf{14}, 57--59 (1965).
	\newblock \doi{10.1103/PhysRevLett.14.57}.
	\newblock \urlprefix\url{https://doi.org/10.1103/PhysRevLett.14.57}
	
	\bibitem{Pfenning:2002}
	Pfenning, M.J.: Quantum inequalities for the electromagnetic field.
	\newblock Phys. Rev. D (3) \textbf{65}(2), 024009, 13 (2002).
	\newblock \doi{10.1103/PhysRevD.65.024009}.
	\newblock \urlprefix\url{https://doi-org/10.1103/PhysRevD.65.024009}
	
	\bibitem{Pfenning:2009}
	Pfenning, M.J.: Quantization of the {M}axwell field in curved spacetimes of
	arbitrary dimension.
	\newblock Classical Quantum Gravity \textbf{26}(13), 135017, 20 (2009).
	\newblock \doi{10.1088/0264-9381/26/13/135017}.
	\newblock \urlprefix\url{https://doi-org/10.1088/0264-9381/26/13/135017}
	
	\bibitem{PfenningFord:1997}
	Pfenning, M.J., Ford, L.H.: The unphysical nature of ``warp drive''.
	\newblock Classical Quantum Gravity \textbf{14}(7), 1743--1751 (1997).
	\newblock \doi{10.1088/0264-9381/14/7/011}.
	\newblock \urlprefix\url{https://doi.org/10.1088/0264-9381/14/7/011}
	
	\bibitem{Synge1967}
	Synge, J.L.: Timelike helices in flat space-time.
	\newblock Proc. Roy. Irish Acad. Sect. A \textbf{65}, 27--42 (1967) (1967).
	\newblock \urlprefix\url{http://www.jstor.org/stable/20488646}
	
	\bibitem{Unruh:1976}
	Unruh, W.G.: Notes on black-hole evaporation.
	\newblock Phys. Rev. D \textbf{14}, 870--892 (1976).
	\newblock \doi{10.1103/PhysRevD.14.870}.
	\newblock \urlprefix\url{https://link.aps.org/doi/10.1103/PhysRevD.14.870}
	
	\bibitem{YuWu:2004}
	Yu, H., Wu, P.: Quantum inequalities for the free {R}arita-{S}chwinger fields
	in flat spacetime.
	\newblock Phys. Rev. D (3) \textbf{69}(6), 064008, 7 (2004).
	\newblock \doi{10.1103/PhysRevD.69.064008}.
	\newblock \urlprefix\url{https://doi-org/10.1103/PhysRevD.69.064008}
	
\end{thebibliography}
%}
\end{document}